%% file: main.tex
\documentclass[
twocolumn,pra,
superscriptaddress,
floatfix,
 amsmath,amssymb,
]{revtex4-2}
\usepackage[pdfencoding=auto, psdextra]{hyperref}
\usepackage{quantikz}
\usepackage{placeins}
\usepackage{amsthm}
\usepackage[caption=false]{subfig}
\usepackage{algorithm2e}
\usepackage{algorithmic}
\usepackage{graphicx}% Include figure files
\usepackage{xcolor}
\input{colors}

\usepackage{orcidlink}

\newtheorem{theorem}{Theorem}[section]

\newtheorem{lemma}[theorem]{Lemma}
\newcommand{\vc}{\bm}
\begin{document}

%\preprint{AIP/123-QED}

\title[]{Fast Quantum Amplitude Encoding of Typical Classical Data}% 

\author{Vittorio Pagni \orcidlink{0009-0006-9753-3656}}
\affiliation{Institute of Software Technology, German Aerospace Center (DLR), Sankt Augustin, Germany}%Lines break automatically or can be forced with \\
 \affiliation{ University of Cologne, Cologne, Germany}
\author{Sigurd Huber \orcidlink{0000-0001-7097-5127}}
\affiliation{Microwaves and Radar Institute, German Aerospace Center (DLR), Oberpfaffenhofen, Germany}
\author{Michael Epping \orcidlink{0000-0003-0950-6801}}
\affiliation{Institute of Software Technology, German Aerospace Center (DLR), Sankt Augustin, Germany}
\author{Michael Felderer \orcidlink{0000-0003-3818-4442}}
 \affiliation{Institute of Software Technology, German Aerospace Center (DLR), Sankt Augustin, Germany}%Lines break automatically or can be forced with \\
 \affiliation{ University of Cologne, Cologne, Germany}
\date{\today}% It is always \today, today,

\begin{abstract}
We present an improved version of a quantum amplitude encoding scheme that encodes the $N$ entries of a unit classical vector $\vc{v}=(v_1,..,v_N)$ into the amplitudes of a quantum state.
Our approach has a quadratic speed-up with respect to the original one. We also describe several generalizations, including to complex entries of the input vector and a parameter $M$ that determines the parallelization.
The number of qubits required for the state preparation scales as $\mathcal{O}(M\log N)$.
The runtime, which depends on the data density $\rho$ and on the parallelization paramater $M$, scales as $\mathcal{O}(\frac{1}{\sqrt{\rho}}\frac{N}{M}\log (M+1))$, which in the most parallel version ($M=N$) is always less than $\mathcal{O}(\sqrt{N}\log N)$.

By analysing the data density, we prove that the average runtime is $\mathcal{O}(\log^{1.5} N)$ for uniformly random inputs. We present numerical evidence that this favourable runtime behaviour also holds for real-world data, such as radar satellite images.
This is promising as it allows for an input-to-output advantage of the quantum Fourier transform.
\end{abstract}

%Use showkeys class option if keyword
\maketitle

\section{Introduction}
Quantum computers have the potential to provide significant speedups over many classical algorithms, such as unstructured search, prime factorization, and the discrete Fourier transform. In particular, the quantum Fourier transform (QFT) offers exponential speedup compared to its classical counterpart \cite{grover1996fast, Shor_1997, coppersmith, Camps_2020}.

Although the theoretical advantages of these approaches have been demonstrated, applying efficient quantum algorithms to classical data requires first encoding it into a quantum state. The lack of an efficient state preparation algorithm can become a computational bottleneck, undermining any improvement. This challenge is particularly relevant in applications that require large classical input vectors, such as machine learning and satellite image analysis.

In this work, we consider the analysis of spectral properties of large images captured by synthetic aperture radar (SAR), see Fig.~\ref{initial_image} as an example of an engineering application that can benefit from the QFT.
\begin{figure}[tbp]%
    \centering%
    \reflectbox{\rotatebox[origin=c]{180}{\includegraphics[width=8cm]{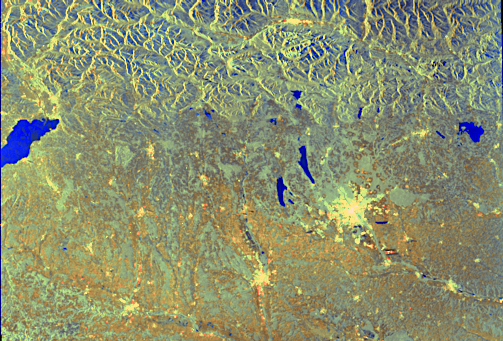}}}%
    \caption{Rescaled image of a part of Germany taken by the satellite Sentinel-1A using a SAR interferometric radar
instrument \cite{CopernicusSentinel1}. Sentinel-1 is a European radar mission developed for the Copernicus program, a joint initiative of the European Commission (EC) and the European Space Agency (ESA).}%
    \label{initial_image}%
\end{figure}%
SAR is a radar technology that uses microwaves to produce detailed maps of objects or even entire planetary surfaces \cite{Curlander_1991}. There is a rich body of literature on various processing techniques, including the works of Soumekh \cite{Soumekh_1994} and Cumming \cite{Cumming_2005}. Many SAR focusing algorithms, such as the $\omega$-$k$ algorithm \cite{78293, Waller_Keil_Friederich_2023} or the range-Doppler algorithm \cite{158864}, rely on the discrete Fourier transform (DFT) and could, in principle, be accelerated by the QFT. For regular SAR imagery, this involves transforming two-dimensional, complex-valued SAR raw data into the spectral domain. Depending on the frequency band and SAR system specifications, a single dataset can consist of samples on the order of $10^{9}$, corresponding to approximately 15 GB of complex double-precision data. Each sample, which is usually stored as a 2D array, can be flattened as a 1D array.

As previously mentioned, encoding such large input vectors into a quantum state is a challenging task, as it requires an efficient generic state preparation algorithm. In particular, the circuit depth and classical runtime (the time that is needed to find the correct circuit) must scale slowly enough to preserve the quantum speed up.
State preparation of quantum states is a fast developing field of quantum computing, with important recent works such as \cite{shukla2023efficient} where the authors introduce an algorithm to efficiently prepare a uniform superposition of the first $1\leq K\leq 2^n$ basis states in the computational basis with an $\mathcal{O}(\log(K))$ circuit depth
\begin{equation}
\ket{\Psi_K}=\frac{1}{\sqrt{K}}\sum_{j=0}^{K-1}\ket{j}.
\end{equation}
Even though this procedure cannot be used to prepare an arbitrary state, it provides an efficient building block for other algorithms and gives some insight on what makes certain state more difficult to prepare than others.

In \cite{Gleinig} the authors present an efficient algorithm for the preparation of sparse states with $m, 0 \leq m\leq N$ non-zero entries which has a runtime of $\mathcal{O}(m^2 \log(m)n)$ and requires $\mathcal{O}(m n)$ qubits.

In \cite{UniformQuantumStatesBoolean} Mozafari et al. introduce an alternative approach based on (reduced) decision diagrams which outperformed state-of-the-art algorithms at preparing generic n-qubit states with $m=n^3$ non zero amplitudes.
In their work, the difficulty of preparing a state, expressed in terms of circuit complexity, is linearly proportional to the amount of paths in the decision diagram corresponding to it. This number is upper-bounded by the sparsity m (the amount of non zero amplitudes) and therefore the algorithm is efficient at preparing sparse states. Moreover, even certain non-sparse states that can be described by a simple decision tree exhibit a relatively small circuit depth.

In this paper we present an improved version of a protocol that encodes the real or more generally complex entries of a classical unit vector $\vc{v}=(v_1,..,v_N)$ of size $N$ into the amplitudes of a quantum state
\begin{equation}
\label{amplitude_encoding_state}
    \ket{\Psi}=\sum_{i=0}^{N-1}{w_i \ket{i}},
\end{equation}
where $n=\log_2(N)$ is the number of qubits and the entries of the vector $\vc w$ approximate those of the input vector $\vc v $ up to a fixed binary precision $L \in \mathbb{N}$.

The key element of the procedure, which is based on the abstract description given by Sahel Ashhab \cite{PhysRevResearch.4.013091}, is the encoding circuit. After an efficient pre-processing that prepares the classical data for a faster encoding, this circuit turns the entries of the input vector into quantum amplitudes by means of controlled rotations.

The depth of the circuit depends on the data density \cite{PhysRevResearch.4.013091}
\begin{equation}
\rho(\vc{w})=\frac{1}{N}\sum_{i=0}^{N-1}{\left(\frac{w_i}{\|w\|_\infty}\right)^2}
\label{eq:definitionofrho}
 \end{equation}
 of the (approximated) input vector, where $\|\vc w\|_\infty=\underset{i}{\max}|w_i|$ is the max norm.
 Note that the data density is bounded by 
\begin{equation}
\frac{1}{N}   \leq \rho \leq 1.
\end{equation}

In Section~\ref{main_results}, we introduce our version of the encoding algorithm, highlighting the most relevant differences from previous approaches and we then analyze the runtime scaling. At the end of this section we show the results of our case study, where we investigate the scaling of the local data density $\rho$ in the EOWEB GeoPortal dataset. In Section~\ref{Discussion} we discuss the limitations and strengths of our approach.
This is followed by a description of the main components of the encoding circuit and their implementation in Section~\ref{methods}.
The proof of the runtime scaling, the state evolution and additional details can be found in the supplementary material. 

\section{Results} 
\label{main_results}
\subsection{An improved state preparation algorithm that uses amplitude amplification for post-processing}
The encoding block that we use is based on \cite{PhysRevResearch.4.013091}.
 To trade memory for execution time, we introduce a parameter M, which determines how many entries of the input vector are encoded in parallel. This can be useful for small devices, where memory is the limiting factor.
The original encoding procedure has a success probability which is equal to the density $\rho$ of the encoded vector. The system register contains the correct superposition only if we measure the value 1 on a specific ancilla qubit, which we call the FLAG, whereas for 0 we have to repeat the encoding process from the beginning. This means that we have to run the circuit and perform a measurement on average $\frac{1}{\rho}$ times before we obtain the desired encoding.

In our circuit we use the amplitude amplification algorithm to boost the amplitude related to the state which carries the right superposition. As usual when applying Grover's operator, which is the core step that is repeated in the amplitude amplification algorithm \cite{Brassard_2002}, see supplementary material~\ref{AppendixD}, the procedure requires a number of iterations $\propto \frac{1}{\sqrt{\rho}}$ to provide the right superposition with a probability greater than $\frac{1}{2}$.

We now give details for the algorithm, see also Algorithm~\ref{shallow_amplitude_encoder}, and describe its action for a real input vector of length N. As we will show in Section~\ref{generalizationToCompl}, this is not a limitation since the procedure can be extended to complex vectors with the same runtime.

To simplify the discussion, we will assume that $N$ is a power of 2 and define 
\begin{equation}
    n=\log_2(N).
\end{equation}
The algorithm consists of three main blocks:
\begin{enumerate}
\item Classical pre-processing
\item Quantum encoding circuit
\item Quantum amplitude amplification
\end{enumerate}
The pseudo code is presented in Algorithm~\ref{shallow_amplitude_encoder}.
\begin{algorithm}[!htbp]
\caption{Shallow Amplitude Encoder}
 \label{shallow_amplitude_encoder}% The number [1] adds line numbers
\flushleft\textbf{Input:}
\begin{itemize}
 \item  $\vc{v} \in \mathbb{R}^N$  vector to encode. This vector, as any other digital data point, is already stored as an (approximated) binary representation in the classical memory. Here we assume the precision of this initial representation to be much larger than the one of the ouput.
\item $L \in \mathbb{N}$ the amount of bits in the binary encoding of each entry of the angle vector $\theta$ introduced below in Section~\ref{classical_preprocessing}. It is directly related to the accuracy of the approximation of the entries of $\vc{v}$ by the amplitudes in the final output state.
\item $M \in \{ 1, .. N\}$  number of encoding operations to perform in parallel at each repetition of the encoding step. The most parallel and therefore the fastest and most memory consuming version of the algorithm corresponds to $M=N$.
\item  $ \mathcal{O}(nM)$ qubits initialized in the $\ket{0} $ standard basis state, respectively for the SYS and for the other registers, that are used as ancillas.
\end{itemize}
\textbf{Output}:\begin{itemize}
 \item $\ket{\Psi}=\sum_{i=0}^{2^n-1}{w_i \ket{i}}$ the n-qubit state whose amplitudes encode the entries of the approximated input vector $\vc w$ with precision at least $2^{-L}$.
\end{itemize}
\begin{algorithmic}[1]
\STATE $B \gets\text{vecToBinaryMatix}(\vc{ v},L)$
\STATE $\mathcal{E}\gets\text{EncodingCircuit}(B,M)$
\STATE $n_\text{Ampl}\gets\left\lfloor \frac{\pi}{4 \cdot \arcsin\left(\sqrt{\rho}\right)} \right\rfloor
$
\STATE $\ket{\Psi}_\text{SYS}\ket{0}_\text{Anc}\gets\mathcal{E}(S_0\mathcal{E}^+S\mathcal{E})^{n_\text{Ampl}}\ket{0}_\text{SYS,Anc}$
\STATE Discard (measure) all the ancilla qubits and keep the first $n$.  
    \STATE \RETURN $\ket{\Psi}_\text{SYS}$
\end{algorithmic}
\end{algorithm}
\subsection{Classical pre-processing}
\label{classical_preprocessing}
\begin{algorithm}[!tpbh]
\caption{vecToBinaryMatrix \label{classical pre-rocessing}}
\flushleft\hspace*{\algorithmicindent}\textbf{Input:}
\begin {itemize} 
\item A real vector $\vc{v} \in \mathbb{R}^N$, $N=2^n$
\item the encoding precision $L \in \mathbf{N}$
\end{itemize} 
\hspace*{\algorithmicindent} \textbf{Output:} a binary matrix $B\in M_{N,L}(\mathbb{Z}_2)$.
\begin{algorithmic}[1] % The number [1] adds line numbers
    \FOR{each $v_i \in \vc{v}$}
        \STATE $\theta_i \coloneqq \arcsin \left( \frac{v_i}{\|\vc{v}\|_\infty} \right) \bigg/ \frac{\pi}{2}$
        \STATE set $B_{i,j} $ such that $\theta_i=(-1)^{B_{i,0}}\sum_{j=1}^{L-1} 2^{-j} B_{i,j} $
    \ENDFOR
    
    \RETURN B
\end{algorithmic}
\end{algorithm}
The first step consists of a classical pre-processing, see  Algorithm~\ref{classical pre-rocessing}. It computes a vector $\vc w$ which approximates the input vector $\vc v$ with a given precision $L \in \mathbb{N}$, and encodes into a binary matrix $B \in M_{N,L}(\mathbb{Z}_2)$, that is calculated via an intermediate angle vector $\vc{\theta}$ with entries
\begin{equation}
   \theta_i := \frac{2}{\pi}\arcsin \left( \frac{v_i}{\|\vc{v}\|_\infty} \right).
\end{equation}
Each row $B_i,i\in \{0,\dots,N-1\}$ is the binary encoding of the corresponding entry $\theta_i$ of the angle vector $\vc{\theta}$ associated to $\vc{v}$
\begin{equation}
    \theta_i=(-1)^{B_{i,0}}\sum_{j=1}^{L-1} 2^{-j} B_{i,j} .
\end{equation}
This operation is completely parallelizable over the entries of the input vector.

We illustrate the pre-processing with the following example, for which we choose a vector $\vc v$ of length $n=3$ and $L=6$,
\begin{equation}
\label{example_vec}
    \vc{v}=\frac{1}{\sqrt{20}}(1,2,-1,2,-1,2,1,2),
\end{equation}
 then the angle vector $\vc{\theta}$ and the binary matrix $B$ are 
\begin{equation}
    \vc{\theta}=(\frac{1}{3}, 1, -\frac{1}{3}, 1, -\frac{1}{3}, 1, \frac{1}{3}, 1) 
\end{equation}
and
\begin{equation}
B = \begin{bmatrix}
0 & 0 & 1 & 0 & 1 & 1 \\
0 & 1 & 1 & 1 & 1 & 1 \\
1 & 0 & 1 & 0 & 1 & 0 \\
0 & 1 & 1 & 1 & 1 & 1 \\
1 & 0 & 1 & 0 & 1 & 0 \\
0 & 1 & 1 & 1 & 1 & 1 \\
0 & 0 & 1 & 0 & 1 & 1 \\
0 & 1 & 1 & 1 & 1 & 1
\end{bmatrix}.
\end{equation}

\subsection{Quantum encoding circuit \texorpdfstring{$\mathcal{E}$}{E}}

The key component of the algorithm is the circuit block $\mathcal{E}$ which performs the (approximated) encoding of each entry $v_i$ as the amplitude of the correspondent computational basis state $ \ket{i}$ in the final superposition, see Eq.~(\ref{amplitude_encoding_state}).

\subsubsection{Structure}
\begin{figure*}[tbp]
    \centering
    \input{circuit_tex_files/encoding_circuit_latex}
    \caption{Visual representation of the encoder circuit $\mathcal{E}$. It is divided into the initial step, which is performed only once, the encoding step that is repeated an amount of times $\frac{N}{M}$, and the final step.}
    \label{circuit_figure}
\end{figure*}
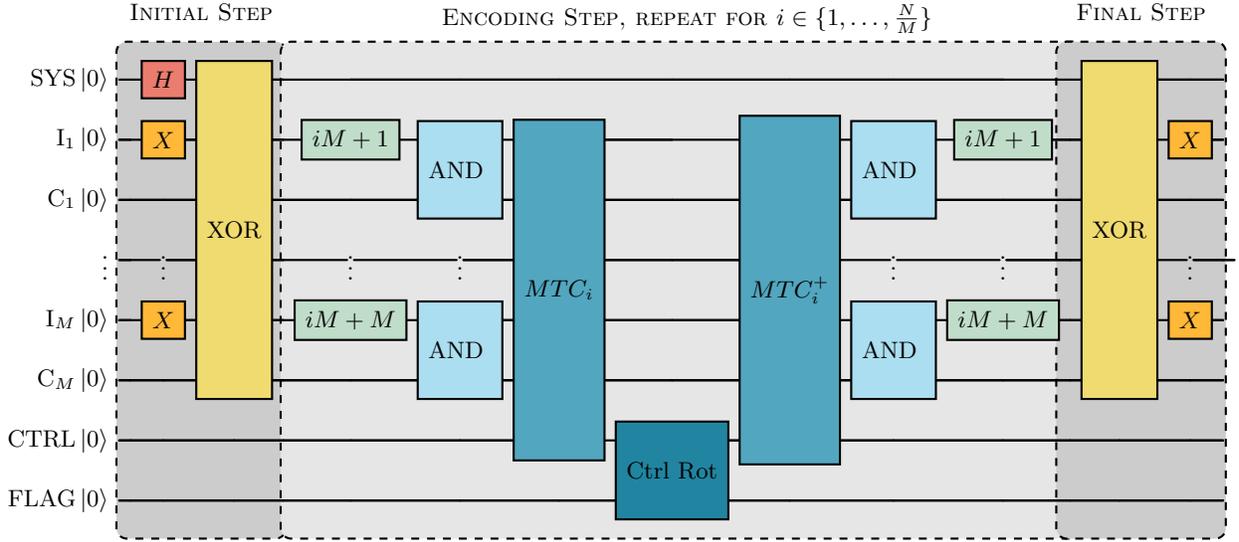

The structure of the encoding circuit presented in Fig.~\ref{circuit_figure} can be divided into three main steps:
an initial step, which is performed only once for each encoding operation independently of the choice of parallelization parameter $M$. It has the role of preparing a uniform superposition of basis states on the SYS register and entangling it with the ancilla registers using a series of C-NOTs (Fan-Out) that act as a parallel XOR operation between the SYS register and each ancilla Index register. 

This step is followed by the actual parallel encoding of $M$ entries of the vector at a time. For the sake of a clear presentation, we assume that $N=M\cdot R$ with $R,M,N \in \mathbb{N}$, such that this operation is repeated $R$ times. 

Finally, there is a third step that disentangles the SYS and the FLAG registers from all the ancilla registers (uncomputation).

\subsubsection{Registers}
There are seven different quantum registers: the most relevant one is the output register SYS, which will carry the right superposition that represents the amplitude encoding of $\vc{v}$. The other six types are just ancilla registers that are necessary to perform the encoding of different amplitudes in parallel.
They are described in Table~\ref{register_table}.
\begin{table}[tbph]
\label{register_table}
    \centering
    \begin{tabular}{|c|c|c|p{3.5cm}|}
    \hline
    Name & Size & Amount& Description \\[0.5ex]
    \hline
    SYS & $n$ & $1$ &The output register that stores the amplitude encoding. \\
    \hline
    Index  & $n$ &$M$& Initialized to the basis state corresponding to their index. \\
    \hline
    Parity compression & $1$ &$M$ &Stores the parity information for compression. \\
    \hline
    CTRL & $L$ &$1$& Control register used in the fan-in operation. \\
    \hline
    FLAG & 1 &1& Target of the controlled rotations; indicates whether the encoding was successful. \\
    \hline
    \end{tabular}
    \caption{Register details and their descriptions. The total amount of qubits is $N_\text{Tot}=n(1+M)+M+L+1+\mathcal{O}(M)=\mathcal{O}(nM)$.}
\end{table}
The circuit also requires an additional amount of ancilla qubits that do not change the $\mathcal{O}(Mn)\leq \mathcal{O}(Nn)$ scaling of the memory but that can vary significantly with the implementation of each logical step of the circuit in terms of their decomposition, as shown in more detail in Section~\ref{methods}.
\subsubsection{Gates}
The circuit uses simple Pauli gates such as $X$ and $iM+1$, the Hadamard gate $H$ as well more complex ones from the set
\begin{equation}
    \mathcal {G}=\{\text{XOR}, \text{AND},\text{MTC}_i,\text{CTRL Rot}\}.
\end{equation}
A description of each gate in this set can be found in Section~\ref{methods}.

\subsection{State Evolution}
\label{state_evolution}
The main goal of the encoding circuit is to use controlled $R_y$ rotations, controlled by the CTRL register and targeting the FLAG, in order to associate to each basis state $\ket{i}$ the correct amplitude $\vc {w}_i$. This operation is performed in superposition and in two main steps: first a uniform superposition of all the basis states is prepared, where the parity compression register $C_1,\dots,C_N$ has only the relative sub-register $C_i$ in the one state, then the controlled rotations turn the uniform superposition of basis states into a weighted one, where the weights are the entries of $\vc{v}$.
The uniform superposition 
\begin{equation}
\begin{aligned}
\label{uniform}
\ket{\Psi_\text{uniform}}=&\frac{1}{\sqrt N}\sum_{k=0}^{N-1} \{ \ket{k}_{\mathrm{SYS}}\bigotimes_{j=0}^{N-1}\ket{1^{\otimes n}\oplus k \oplus j}_{I_j}\ket{\delta_{k,j}}_{C_j}\\
&\bigotimes_{r=0}^{L-1}\ket{B_{k,r}}_{\mathrm{CTRL}}\ket{0}_{FLAG}\}
\end{aligned}
\end{equation}
becomes 
\begin{align}
\ket{\Psi_\text{weighted}}=&\frac{1}{\sqrt N}\sum_{k=0}^{N-1} \{ \ket{k}_{\mathrm{SYS}}\bigotimes_{j=0}^{N-1}\ket{1^{\otimes n}\oplus k \oplus j}_{I_j}\ket{\delta_{k,j}}_{C_j}\nonumber\\
&\bigotimes_{r=0}^{L-1}\ket{B_{k,r}}_{\mathrm{CTRL}}\big((-1)^{B_{k,0}} \sqrt{1 - c_k^2}\ket{0}_{\mathrm{FLAG}}+\\
&+c_k\ket{1}\big)_{\mathrm{FLAG}}\}\label{weighted_sup}\nonumber
\end{align}
which can be written as
\begin{equation}
\ket{\Psi_\text{weighted}}=\sqrt{1-\rho}\ket{\Psi_{(0)}}\ket{0}_{\mathrm{FLAG}}+\sqrt{\rho}\ket{\Psi_{(1)}}\ket{1}_{\mathrm{FLAG}},
\end{equation}
where
\begin{equation}
\begin{aligned}
\ket{\Psi_{(1)}}=&\sum_{k=0}^{N-1}w_k \ket{k}_{SYS}\bigotimes_{j=0}^{N-1}[\ket{1^{\otimes n}\oplus k \oplus j}_{I_j}\\
&\ket{\delta_{k,j}}_{C_j}]\bigotimes_{r=0}^{L-1}\ket{B_{k,r}}_{CTRL}
\end{aligned}
\end{equation}
is almost the desired encoding state, except for the fact that it is still entangled with the ancilla registers. The rest of the operations separate these registers from the SYS and the FLAG, giving the final state
\begin{equation}
\label{weighted_superposition}
\ket{\Psi_\text{weighted}}=\sqrt{1-\rho}\ket{\Psi_{B}}\ket{0}_{FLAG}+\sqrt{\rho}\ket{\Psi_{G}}\ket{1}_{FLAG},
\end{equation}
where 
\begin{equation}
\ket{\Psi_{G}}=\sum_{k=0}^{N-1}w_k \ket{k}_{SYS}\bigotimes_{j=0}^{N-1}[\ket{0}_{I_j}
\ket{0}_{C_j}]\bigotimes_{r=0}^{L-1}\ket{0}_{CTRL}.
\end{equation}
This superposition tells us that whenever we measure 1 on the FLAG register, which happens with probability $\rho$, we have prepared the correct amplitude encoding. Conversely, whenever we measure the value 0, we have to repeat the procedure from the beginning. This behavior has already been described in \cite{
PhysRevResearch.4.013091}. In this work we boost the probability of measuring the correct state by means of amplitude amplification \cite{Brassard_2002}, see supplementary material~\ref{AppendixD}.   \subsection{Theoretical scaling of runtime and memory cost \label{section5}} 
Even though the final desired amplitude encoding state is stored on the register SYS, the total amount of ancillas to perform the algorithm in parallel scales as $\mathcal{O}(Mn)\leq \mathcal{O}(Sn)\leq \mathcal{O}(Nn)$.

The main contribution in this sense is given by the M Index registers, each of size n. This number is upper bounded by the amount of non zero entries S in the vector after the $2^{-L}$ precision approximation. 

In supplementary material~\ref{AppendixA} we show how the memory cost can be reduced significantly in case of redundant or sparse input vectors.
It is important to note that several other approaches in the literature can exploit this constraint on the structure of the input, for example in \cite{UniformQuantumStatesBoolean} they analyse the decision diagram associated to the target state, which simplifies significantly when we impose certain constraints on the input.

\label{timeCost}
The encoding block of the algorithm has depth
\begin{equation}
\label{full_scaling_N_M}
    \mathcal{O}(\frac{N}{M}\log{(M+1)})
\end{equation}
and must be applied  $\mathcal{O}(\frac{1}{\sqrt \rho})$ times during the amplitude amplification step. This means that the total depth of the circuit in the most parallel case ($M=N$) and hence the total run time $\tau$ follows 
\begin{equation}
\tau(\vc v)\propto n\frac{1}{\sqrt{ \rho(\vc v)}}=\sqrt{N \|\vc v\|_\infty^2}n,
\end{equation}
where we used the formula for $\rho$ in Eq.~(\ref{eq:definitionofrho}) in the case of a unit vector ($\| \vc v \|_2=1$)
\begin{equation}
\rho(\vc{v})=\frac{1}{N}\sum_{i=0}^{N-1}{\left(\frac{v_i}{\|\vc v\|_\infty}\right)^2}=\frac{1}{N \|\vc v\|_\infty^2}.
\end{equation}
In supplementary material~\ref{AppendixA} we describe in more detail the properties of the squared max norm over $\mathbf{S}_{N-1}$ and we show that, since $\frac{1}{N}\leq \|v\|_\infty^2\leq 1$, in the worst case scenario the runtime scales as
\begin{equation}
\tau _\text{worst}\sim \sqrt{N}n.
\end{equation}
Furthermore, we also prove that for a randomly chosen vector that is uniformly sampled over the N-sphere $\mathbf S_{N-1}$ we have the following behavior for the average and the variance of $\tau$

\begin{align}
\mathbf{E}[\tau(\vc v)\mid \vc v \sim \text{uniform}(\mathbf{S}_{N-1})]\overset{N \gg 1}{=}& \mathcal{O}\left( n^{1.5} \right)\\
\text{and }\text{Var}[\tau(\vc v)\mid \vc v \sim \text{uniform}(\mathbf{S}_{N-1})]\overset{N \gg 1}{=}& \mathcal{O}\left( \frac{n^{1.5}}{\sqrt{N}} \right),
\end{align}
respectively.
In practice, we expect this scaling to hold for large images, for example those in \cite{shermeyer2020spacenet}, and we give an example of this behaviour in Section~\ref{case_study}.
\subsection{Generalization to complex numbers} \label{generalizationToCompl}
In order to use the output of this encoding procedure as initial state for many quantum algorithms, for instance the QFT, it is necessary to extend this framework to complex numbers.

If we consider a complex vector $\vc{v}=(R_1 e^{i\phi_1}, .. R_N e^{i\phi_N})$, we can encode the modulus and the phase of each entry separately, which means using different registers, but simultaneously.

Again we require that $\|\vc{v}\|= \sqrt{\sum_{i=0}^{N-1}{R_i^2}}=1$
The modulus can be encoded using the same circuit described above, while the phase of each entry could be encoded by repeatedly applying controlled phase gates to FLAG registers.
We use $2N$ registers of n qubits to encode respectively the modulus and the phase angle of each entry.
The pre-processing is analogous to the real case.

Given a complex vector $\vc{v}=(R_0 e^{i\phi_0}, .., R_{N-1} e^{i\phi_{N-1}})$ we define 
the modulus vector 
\begin{equation}
    \vc{R}=(R_0,R_1,\dots,R_{N-1}),
\end{equation}
normalized with $\sum_{j=0}^{N-1}R_j^2=1$,
and the phase vector 
\begin{equation}
    \vc{\phi}=(\phi_0,\phi_1,\dots,\phi_{N-1}),
\end{equation}
where we fix the first phase to $\phi_0=0$ using the invariance under global phase shift of the quantum state.
From the modulus vector we go to the angles vector defined as

\begin{equation}
\theta^{(R)}_i:=\frac{2}{\pi}\arcsin(\frac{R_i}{\|R\|_\infty}),\qquad \theta^{(R)}_i \in [0,1]
\end{equation}
while for the phase we simply define 
\begin{equation}
\theta^{(\phi)}_i=\frac{\phi_i}{\pi},\qquad \theta^{(\phi)}_i \in [-1,1]
\end{equation}
If we encode $\vc{\theta}^{(R)}$ and $\vc{\theta}^{(\phi)}$ into two separate binary matrices $B^{R},B^{\phi} \in \mathbf{M}_{N \times L}(\mathbb{F}_2)$, we can repeat the same procedure described above to encode the modulus register, using a series of controlled $R_y$ rotations, whereas the phase is encoded using $R_s$ gates with decreasing argument, 
\begin{equation}
   R_s= \begin{pmatrix}
1 & 0 \\
0 & e^{\frac{i2\pi}{2^s}}
\end{pmatrix}
\end{equation}
No sign bit is used in the binary encoding of $\vc{\theta^{(R)}}$ as $R_i \geq 0, \forall i \in \{0,N-1\}.$

\subsection{Case Study - Analysis of the data density \texorpdfstring{$\rho$}{rho} over satellite images}\label{case_study}
In this section we present as our case study the analysis of the behaviour of our data density parameter $\rho$ over images from the data base Copernicus Data Space Ecosystem \cite{sentiwiki2025}. In particular, we divide the image into a grid of sectors and we calculate the density of the flattened and normalized vector associated to each sector, see Fig~\ref{heatmaps} and Fig~\ref{sectors}. We repeat this process varying the amount of sectors in the grid and therefore their size and then plot how the average value of the sector density changes.

The analysis aims to showcase a real dataset that could be suitable for the proposed fast encoding approach. This approach could, for instance, enable large images to be loaded into a quantum state and their quantum Fourier transform to be performed in sub-linear time relative to the input size \cite{6129631,application_of_fft_sar,Camps_2020}.
We emphasize that we directly calculate the density of the image data and, in contrast to Section~\ref{section5}, we do not assume a distribution for the input data.  

\subsubsection{Scaling of $\rho$ over a satellite image of Germany}
In this case study we consider a satellite image of Germany shown in Fig.~\ref{initial_image} taken from \cite{CopernicusSentinel1}. 
\begin{figure*}[tpbh]
    \centering
    \subfloat[This heat map shows the local density $\rho$ for a grid with 20 x 20 sectors considering the gray scale version of the merged image. At the center of the blue circle we indicate in green the sample sector displayed in Fig~\ref{sector_20}.]{
        \centering
        \includegraphics[width=0.32\linewidth]{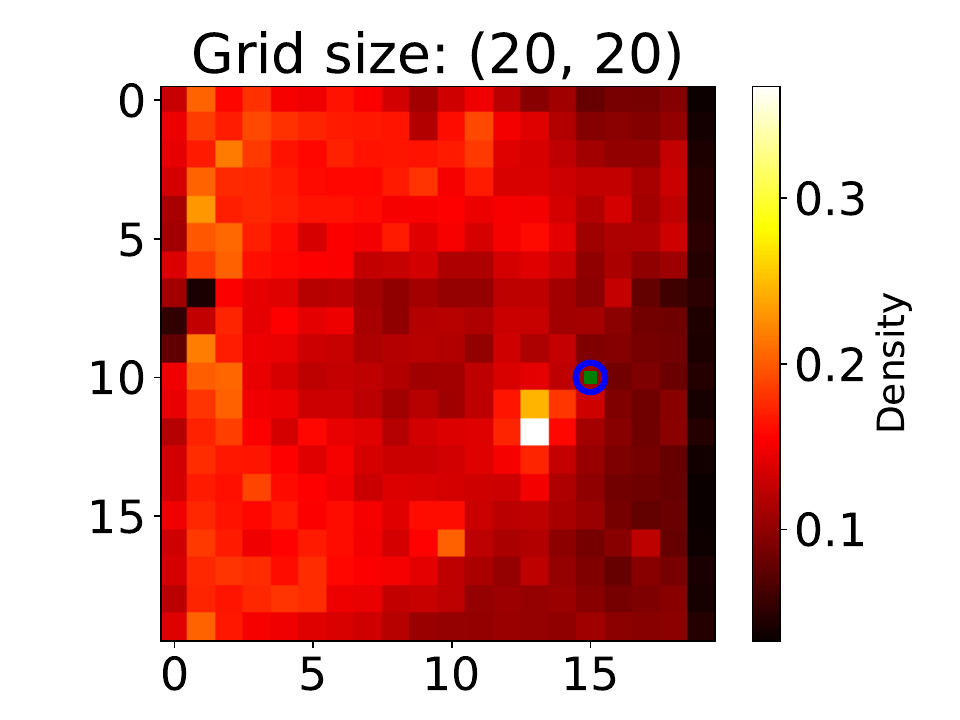}
        \label{heatmaps_20}}
    \hfill
    \subfloat[Heat map of the local density for a grid with 35 x 35 sectors.  At the center of the blue circle we indicate in green the sample sector displayed in Fig~\ref{sector_35}.]{
        \centering
        \includegraphics[width=0.32\linewidth]{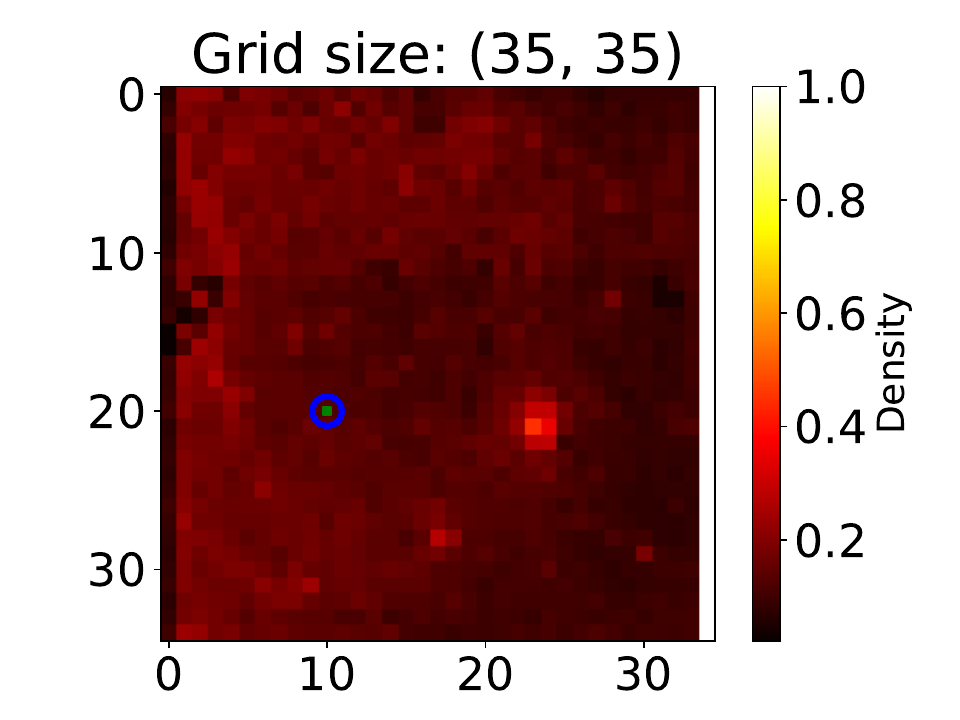}
        \label{heatmaps_35}
    }\hfill
    \subfloat[Heat map of the local density for a grid with 100 x 100 sectors.  At the center of the blue circle we indicate in green the sample sector displayed in Fig~\ref{sector_100}.]{
        \centering
        \includegraphics[width=0.32\linewidth]{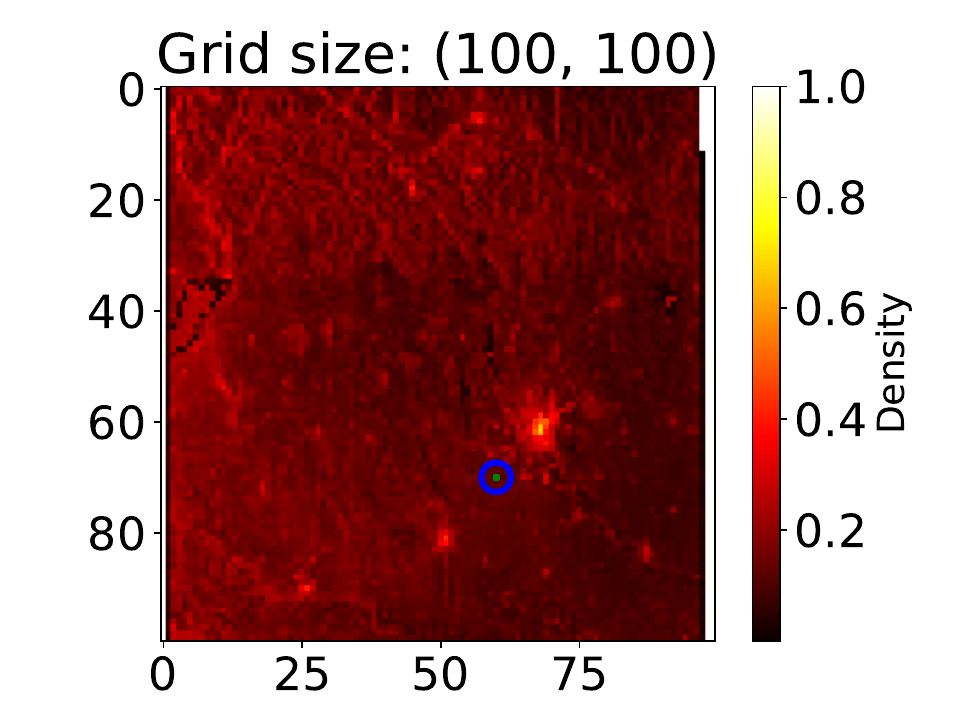}
        \label{heatmaps_100}
    }
    \caption{Heat maps displaying local density for different grid sizes. We notice how the average density increases as the sector size decreases.}
    \label{heatmaps}
\end{figure*}
\begin{figure*}[tpbh]
    \centering
    \subfloat[Sector 15,10 of size (833,1318) from the 20x20 grid in Fig~\ref{heatmaps_20}.]{
        \centering
        \includegraphics[width=0.32\linewidth]{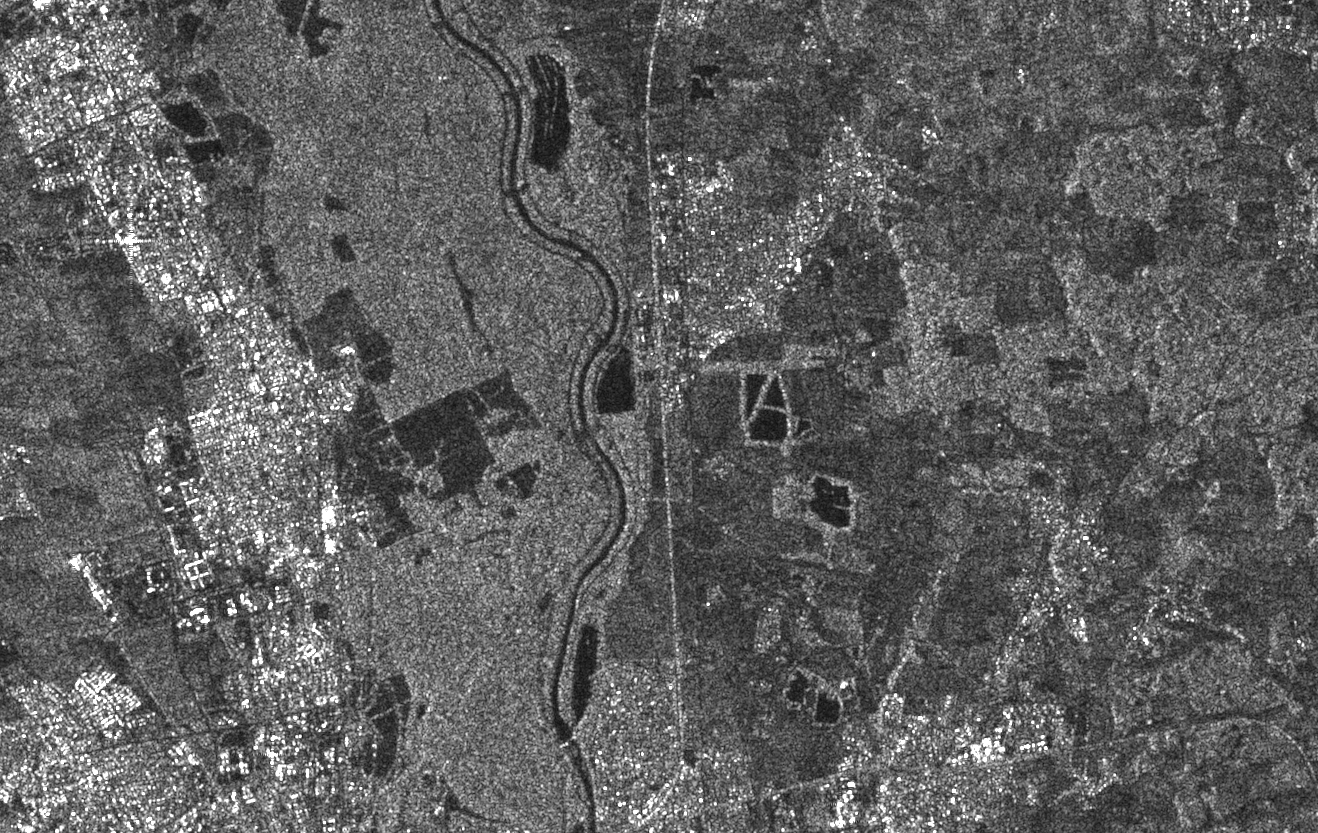}
        \label{sector_20}
    }\hfill
    \subfloat[Sector 10,20 of size (476,753) from the 35x35 grid in Fig~\ref{heatmaps_35}]{
        \centering
        \includegraphics[width=0.32\linewidth]{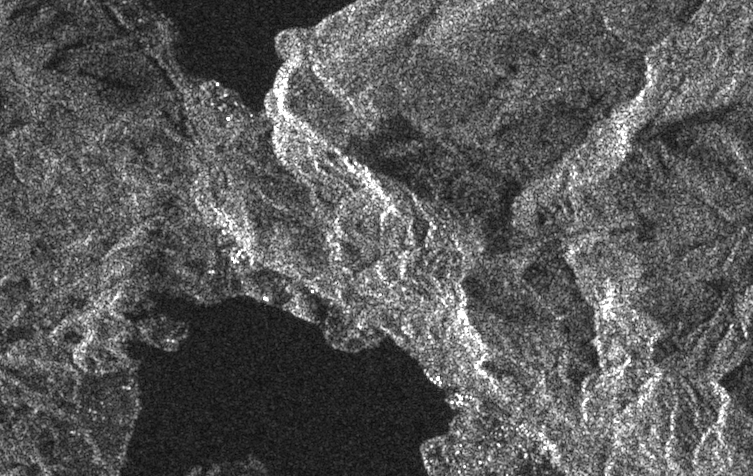}
        \label{sector_35}
    }\hfill
    \subfloat[Sector 60,70 of size (166,273) from the 100x100 grid in Fig~\ref{heatmaps_100}.]{
        \centering
        \includegraphics[width=0.32\linewidth]{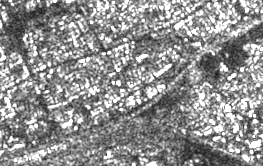}
        \label{sector_100}
    }
    \caption{Examples of different sectors with varying sizes from gray scale version of Fig~\ref{initial_image}, each corresponding to a specific grid size.}
    \label{sectors}
\end{figure*}

Our analysis consists of dividing the image into a grid of an increasing amount $n_s^2$ of sectors, whose size hence decreases as $\frac{n_{\text{pixels}}}{n_s^2} $. We then normalize each sector independently and compute its local $\rho$.
Fig.~\ref{heatmaps} and Fig.~\ref{sectors} respectively show a heat map of the local density for different grid sizes and a sample image taken from the correspondent grid.

We analyze the scaling of the average vector density $\rho$ over the sectors as their size varies and we get the plot showed in Fig.~\ref{density_scaling} for the gray scale version of the image in Fig.~\ref{initial_image}. From the plot we see how the density decreases as we increase the size of the sectors compared to a $\propto \frac{1}{\sqrt{x}} $ and a $\propto \frac{1}{\log{x}} $ scaling, the latter being the average behavior we would expect from a uniformly sampled set of vectors over $\mathbf{S}_{M-1}$, where $M$ is the size of the flattened normalized versions of the 2D arrays of the sectors. As we explain in the supplementary material \ref{averageRuntime}, sampling the image vector uniformly from the hypersphere corresponds to sampling its components independently according to a standard gaussian distribution ($\mu=0, \sigma=1$) and then normalizing the vector to 1.
\begin{figure}
    \centering
    \includegraphics[width=\linewidth]{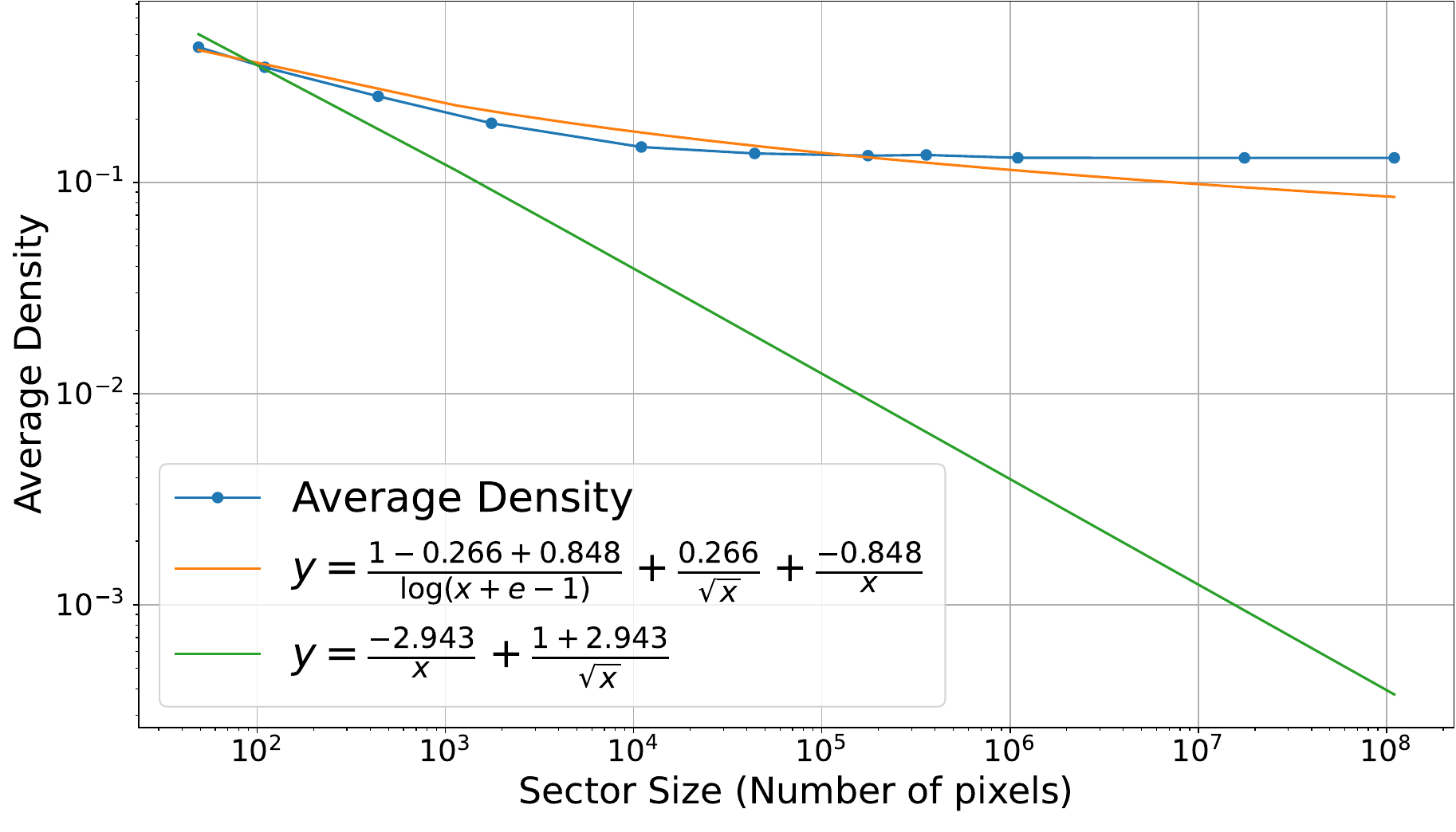}
    \caption{Behavior of the average data density as a function of the size of the flattened and normalized versions of the 2D arrays of the sectors compared to the one expected when sampling from a uniform distribution over the N-sphere and to $\propto \frac{1}{\sqrt{x}}$.}
    \label{density_scaling}
\end{figure}
\subsection{Example - Application to QFT \label{QFT_example}}
One of the possible applications where this approach helps preserving a relevant speed up with respect to a classical algorithm is quantum Fourier transform (QFT).

While the classical Cooley-Tukey algorithm \cite{Cooley1965AnAF} for computing FFT has in general a scaling of $\mathcal{O}(Nn)$, which could potentially get closer to a linear one by parallelizing on multiple processors \cite{scalabilityFFT,Ayala2021} or even sub-linear by performing sparse FFT (under specific assumptions)\cite{Ermeydan2018}, the QFT has a scaling of $\mathcal{O}(n^2)$, see Fig.~\ref{qft_circuit}, which can be pushed to $\mathcal{O}(n \log n)$ if we have controlled rotations as native gates~\cite{QFFT}.

It follows that, if we are able to encode a classical vector of size N with a scaling that is better than linear, we can at least in theory partially preserve the quantum advantage of QFT over FFT on classically stored data.

This result seems particularly useful for quantum machine learning and image processing where large vectors are read, processed and stored multiple times \cite{AnalysisFFTImages,FFTDeepLearning}.

Fourier transform and its variants have several applications related to image processing, for example in \cite{FFT_for_SAR}  the authors describe a SAR method that uses unequally-spaced fast Fourier transforms (USFFTs), while in \cite{SPECAN} momentary Fourier transform (MFT) is used in the Spectral Analysis Algorithm (SPECAN).
\begin{figure}[tbph]
\center
    \includegraphics[width=\linewidth]{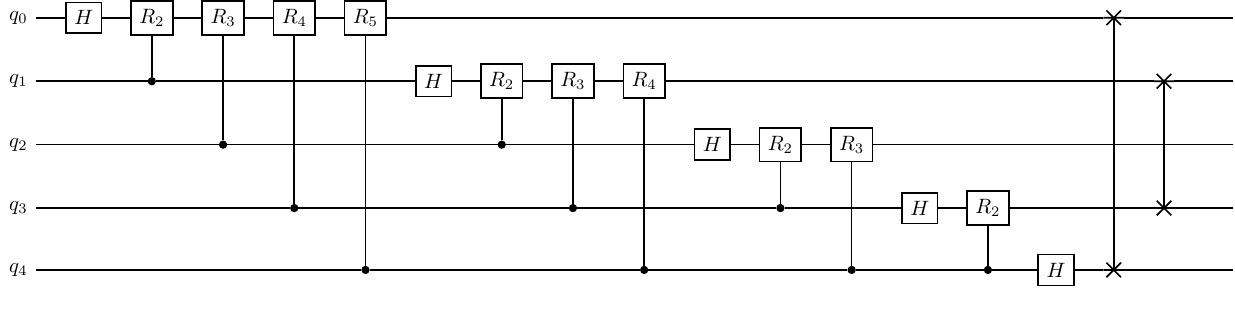}
    \caption{Circuit performing the QFT on 5 qubits.}
    \label{qft_circuit}
\end{figure}
As we mentioned before, if we consider the circuit that encodes and transforms the classical input with controlled rotations as native gates, the $\mathcal{O}(n\log n)$ contribution of the fast QFT to the total runtime can be neglected. This circuit maintains a worst case runtime
\begin{equation}
\mathbf{\tau_\text{QFT}}_\text{Worst}\propto \sqrt{N}n 
\end{equation}
and an average scaling over the uniform distribution on the n-sphere 
\begin{equation}
\label{average_runtime}
\mathbf{E}[\mathbf{\tau}_\text{QFT}]\propto n^{1.5}.
\end{equation}
In order to extract the results of the QFT we need to run and measure the circuit multiple times, in series or in parallel. If we are interested in retrieving every entry of the transformed vector with some fixed precision, this requires $\mathcal{O}(N)$ unentangled measurements \cite{Verdeil_2023}.
This exponential scaling in the amount of qubits comes directly from the complexity of state tomography and the output problem of quantum algorithms. It is important to notice that 
\begin{enumerate}

    \item this large amount of measurements is required only whenever we want to read the full output pure state and store it classically. Some interesting properties of the output, for example the expectation value of an observable on it, can be computed much more efficiently, as in \cite{quantumImageFilteringCui},
    \item the Fourier transformed output can be used as it is as input for another quantum circuit, without performing any measurement and therefore with the average runtime in Eq.~(\ref{average_runtime}).
\end{enumerate}

\section{Discussion}\label{Discussion}

Although the encoding procedure requires an amount of ancilla qubits given by $\mathcal{O}(Mn)$ with $1\leq M \leq S \leq N$, where $M$ is the parallelization parameter and $S$ is the number of nonzero entries in the input vector, the SYS quantum register that stores the final superposition has exponentially smaller size than the original amount of bits in its classical representation. Considering this memory cost together with the runtime scaling in Eq.~\ref{full_scaling_N_M} as a function of $N$ and $M$ , we observe that in order to fast encode generic and possibly unstructured input data with a  shallow circuit, we need to choose the largest possible value for M, at the cost of potentially requiring a superlinear amount of ancilla qubits in the input size. However, varying the parameter $M$ allows us to trade runtime and memory cost, adjusting it to the resources at our disposal.\\
After its preparation, the encoded state can be used as input to other quantum circuits or to their subroutines, such as, for example, the quantum Fourier transform or the quantum k nearest neighbors algorithm,  \cite{Zardini_2024}.

Moreover, in our state preparation algorithm the encoding circuit has a weak dependence on the input vector, allowing the classical pre-processing to be parallelized over the $S$ entries. This enables on-the-fly, sublinear quantum encoding and processing of large vectors. This feature makes this encoding scheme suitable for a cloud service that encodes and processes large input images, for example by applying the QFT and filters as explained in \cite{quantumImageFilteringCui}, and subsequently sends the smaller  encoding state to another device via the quantum internet. This way the sublinear scaling can be preserved even in the worst case scenario $\sqrt{N} n$ and a significantly better average scaling $n^{1.5}$ over a uniformly distributed data set on $\mathcal{S}_{N-1}$ (or, equivalently, over a data set obtained by sampling the entries according to a standard normal distribution and normalizing to one) is achievable.

\section{Methods}
\label{methods}
In this section, we provide an overview of the key components of the encoding circuit, along with a concise explanation of how we utilized amplitude amplification to increase the likelihood of preparing the correct state.
\subsection{Components of the encoding circuit}
The initial step shown in Fig.~\ref{initial_step_fig} involves two logical operations: the preparation of a uniform superposition of N states on the SYS by means of the Hadamard gate $H^{\otimes n}$ and the XOR gate between the SYS and the index registers.
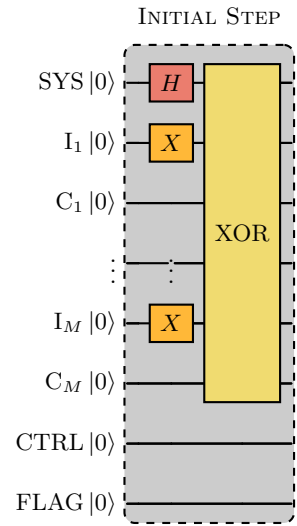
\begin{figure}[tpbh]
    \centering
    \input{circuit_tex_files/initial_step}
    \caption{Initial step of the encoding circuit}
    \label{initial_step_fig}
\end{figure}

\subsubsection{\textbf{The XOR (Fan-out) operation}\label{Fanout}}
The Fan-out operation, described in many QRAM architectures \cite{phalak2023quantum,gokhale2020quantum}, consists of several C-NOTs which share the same control and have different targets, in our case, respectively the SYS and the Index registers, see Fig~\ref{XOR}. 
The action of this gate over the registers is described by 
\begin{equation}
\begin{aligned}
&\text{XOR}[\sum_{i=0}^{N-1}\ket{j}_\text{SYS} \otimes_{i=0}^{N-1}\ket{i+1}_{I_i}]\\
=&\sum_{i=0}^{N-1}\ket{j}_\text{SYS} \otimes_{i=0}^{N-1}\ket{i+j+1}_{I_i}.
\label{Xor_gate}
\end{aligned}
\end{equation}
This operation is well known in the literature and can be efficiently implemented in parallel using GMS gates \cite{Maslov_2018,gokhale2020quantum,van_de_Wetering_2021}.
The GMS gate is described by 
\begin{equation}
\text{GMS}(\chi_{1,2},\chi_{2,3},\chi_{n,n+1})=\exp (-\mathrm{i}\sum_{i=0}^n\sum_{j=i+1}^n X_i X_j \frac{\chi_{i,j}}{2}).
\end{equation}
However, we can provide an alternative parallelization scheme that involves $\mathcal{O}(M)$ ancilla qubits and allows to reduce the circuit depth of this operation to $\mathcal{O}(\log(n))$. This scheme is essentially the same for both Fan-in and Fan-out and is showed in Fig~\ref{Mem2CTRL}.
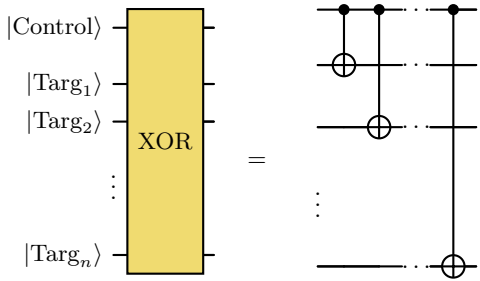
\begin{figure}[tpbh]
    \centering
    \input{circuit_tex_files/XOR}
    \caption{Representation of the XOR operation as a Fan-out.}
    \label{XOR}
\end{figure}

The second section of the circuit, where the actual encoding occurs, is composed of several sub-circuits. It begins with the LoadIndex gate that acts on the Index registers and turns it into the computational basis state associated to the its binary address in the memory. It is followed by the the AND operation between each of the Index registers and the correspondent Parity compression ones, which for each state in the superposition associated with a value $i$ on the system flips the $i$-th parity compression register from the $\ket{0}$ to the $\ket{1}$ state. Then we have the Mem2CTRL gate (Fan-in) between the the Parity compression registers and the CTRL register and finally the controlled rotations that target the FLAG register, which as we will see in more detail in the rest of this section, perform the conversion of a uniform superposition of computational basis states to a weighted superposition, whose weight are the entries of $\vc w$.
\subsubsection{\textbf{The load index operation}}
After performing the XOR operation, we act on each Index register with a tensor product of X and I operators, where the positions of the Xs correspond to those of the ones in the binary representation of the relative index:
\begin{equation}
\text{LoadI}_k:=\bigotimes_{i=0}^n X^{b(k)_i},\qquad k \in \{0,1,\dots N-1\},
\end{equation}
where $b(k)$ is the n-bit binary string representing the integer $k$.

\subsubsection{\textbf{AND (n-Toffoli)}}
The next component of the encoding step is the logical AND operation which is implemented through a n-Toffoli gate, see Fig~\ref{AND}.
Its function is to check if the binary sum of the index and the system registers (which are both added to the parity check register) is zero, leaving the parity check register in the original state $\ket{1^n}$ and therefore indicating a perfect match between the two values.
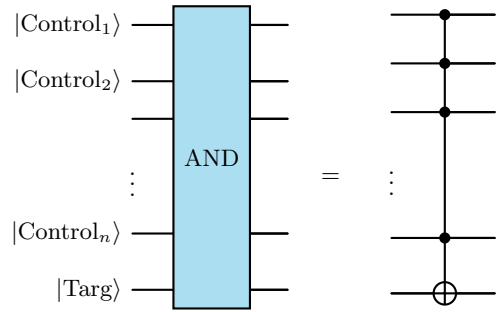
\begin{figure}[tbph]%
    \centering%
    \input{circuit_tex_files/AND}%
    \caption{The representation of the logical AND operation as a  n-Toffoli gate. }%
    \label{AND}%
\end{figure}%
Although in \cite{PhysRevA.101.022308} the authors describe an efficient way to perform an n-Toffoli gate in a single step with high fidelity, at least when n is not too large, in \cite{PhysRevResearch.4.013091} the author provides the recursive decomposition of this kind of gate in terms of $n-1$ ordinary Toffoli gates and using $n-1$ ancilla qubits. 
The interesting thing about this decomposition is that most of the operations involved can be performed in parallel, so that the actual circuit depth of the n-Toffoli gate is $\mathcal{O}(\log n)$.

\subsubsection{\textbf{MTC} as \textbf{Memory2Control} followed by \textbf{CTRL Compression}}
The purpose of the MTC sub-circuit is to load the correct row of the binary matrix B into the CTRL register. This must be done in superposition, so that in the following step this register can be used to control the $R_y$ rotations targeting the FLAG register.
The sub-ciruit is therefore input-dependent, in the sense that the classical information about the entries (in particular of the ones) of the matrix $B$ is used to prepare (or control) a series of C-NOTs that have each a different parity compression register as control and share the CTRL register as their common target. Here the parity compression registers $C_i$ act as a switch that, in superposition, decides which classical bit string to load on CTRL.
\begin{figure}[tbph]%
\centering%
    \input{circuit_tex_files/MTC_i}%
    \caption{Representation of the MTC operation as a Fan-in.}%
    \label{MemToCTRL_figure}%
\end{figure}
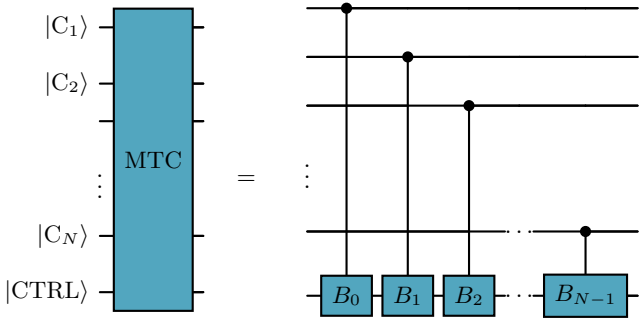%
In Fig.~\ref{MemToCTRL_figure} the operation is shown as serialized, but we can use the same techniques applied to the fan-out, i.e. GMS gates or the ancilla decomposition, in order to parallelize it and make it $\mathcal{O}(n)$ in depth while still maintaining the $\mathcal{O}(Nn)$ memory scaling.
If we choose the second option, this can be achieved by targeting multiple CTRL registers with pairs of C-NOT operations, as shown in Fig~\ref{Mem2CTRL}. The content of these registers is then compressed into a single CTRL register using the sub-circuit described in Fig.~\ref{CTRL_comp}. 
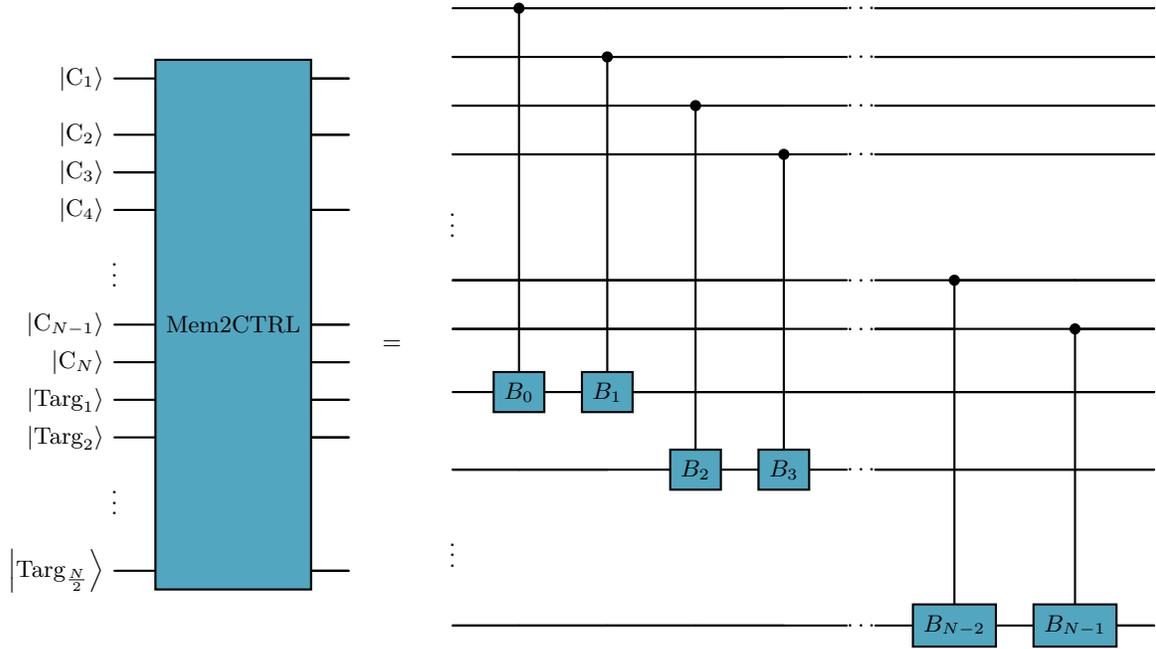
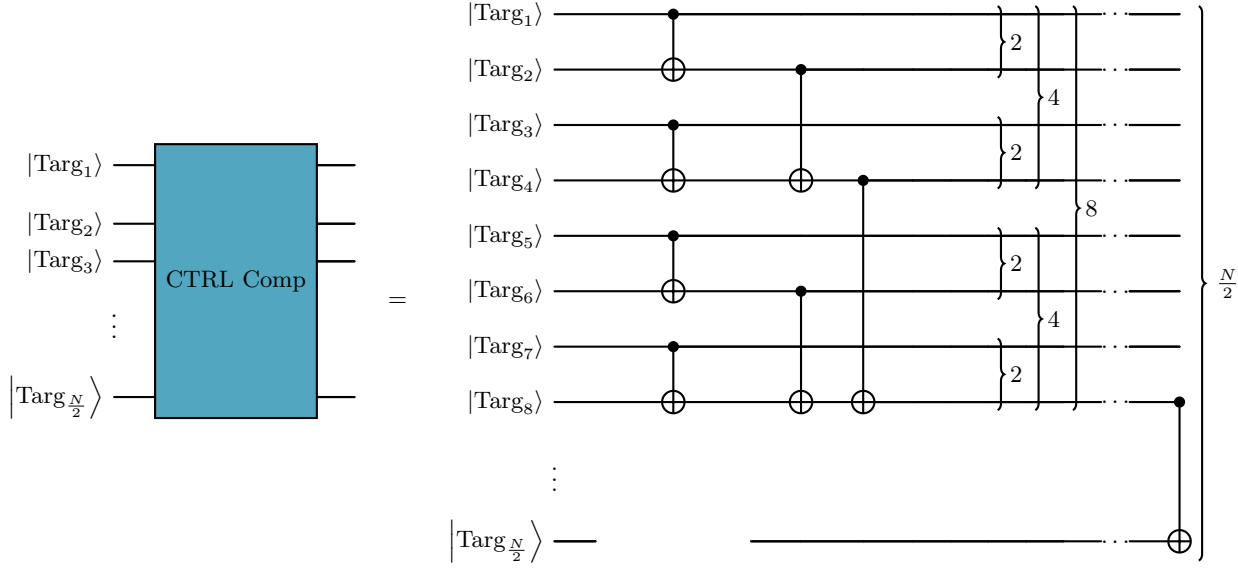
\begin{figure*}[tpb]%
    \centering%
    \subfloat[First part of the parallelization scheme where pairs of classically controlled C-$B_i$ operations act on $\frac{ N }{2}$ ancilla copies of the CTRL register, each requiring L qubits. Each $B_i$ is a tensor product of $X$ gates encoding the ones in the row $\vc B_i$.  This decomposition parallelizes the Fan-Out operation by means of ancilla qubits.]{
        \centering%
         \input{circuit_tex_files/Memory2Control}%
        \label{Mem2CTRL}%
    }%
    \hfill%
    \subfloat[Second part of the parallelization scheme, called CTRL compression, which represents a possible decomposition for the Fan-in. Here, layers of C-NOTs compress the information about the CTRL register using a recursive division approach where each CTRL register has L qubits. The depth of the whole sub-circuit is $\mathcal{O}(n)$, and the memory scaling remains $\mathcal{O}(Mn) \leq \mathcal{O}(Nn)$.]{
        \centering%
        \input{circuit_tex_files/Ctrl_comp}%
        \label{CTRL_comp}%
    }
    \caption{Illustration of the parallelization scheme for the Fan-In.}%
    \label{ctrl_comp}%
\end{figure*}%

\subsubsection{\textbf{Encoding using controlled rotations}}
The binary matrix $B$ associated to the input vector $\vc{v}$ has $N$ rows, one for each entry in $\vc{v}$, and $L$ columns, which correspond to the $L$ bits in the binary representation of the entries of the $\vc{\theta}$ angle vector. 

A key aspect of this encoding is that we can reconstruct each real entry $v_i$ by using $L$ controlled rotations with decreasing angle, see FIG~\ref{controlled_rotations}.
These rotation are either performed or skipped according to the $L$ bits in the correspondent row $\mathbf{B}_i, i \in \{0 \dots N-1\}$ of $B$.
Going back to the example in Eq.~(\ref{example_vec}), the first row is 
\begin{equation}
    \vc{B}_0=(0,0,1,0,1,1)
\end{equation}
if we apply L classically controlled $R_y$ rotations with angles 
\begin{align}
    \phi_0 &= 2\pi\nonumber\\
    \phi_1 &= \pi 2^{-1}\nonumber\\
    \phi_2 &= \pi 2^{-2}\nonumber\\
    &\vdots \nonumber\\
    \phi_{L-1} &= \pi 2^{-(L-1)}
\end{align}
to an ancilla qubit whose initial state is $\ket{0}$ where we use the bits in this row as control we obtain 
\begin{align}
\ket{\Psi_0}=&\prod_{l=0}^{L-1} R_y(\phi_l)^{B_{0,l}}\ket{0}\nonumber\\
=&\hphantom{+{}}\cos(\frac{1}{2}\sum_{l=0}^{L-1}B_{0,l}\cdot \phi_l)\ket{0}+\nonumber\\
&+\sin(\frac{1}{2}\sum_{l=0}^{L-1}B_{0,l}\cdot \phi_l)\ket{1}.
\end{align}
Exploiting the definition of $\phi_i$ and  $\mathbf{B}$ we get 
\begin{align}
\ket{\Psi_0}
=&\hphantom{+{}}\cos(\frac{1}{2}\sum_{l=0}^{L-1}B_{0,l}\cdot \phi_l)\ket{0}+\nonumber\\
&+\sin(\frac{1}{2}\sum_{l=0}^{L-1}B_{0,l}\cdot \phi_l)\ket{1}\nonumber\\
=&\hphantom{+{}}\cos(\frac{1}{2}\sum_{l=1}^{L-1}B_{0,l}\cdot\frac{\pi}{2^l}+ \pi\cdot B_{0,0})\ket{0}+\nonumber\\
&+\sin(\frac{1}{2}\sum_{l=1}^{L-1}B_{0,l}\cdot\frac{\pi}{2^l}+ \pi\cdot B_{0,0})\ket{1}\nonumber\\
=&\hphantom{{}+{}}(-1)^{B_{0,0}}\cos(\arcsin(|\frac{v_0}{v_\infty}|))\ket{0}+\nonumber\\
&+(-1)^{B_{0,0}}\sin(\arcsin(|\frac{v_0}{v_\infty}|))\ket{1}
\end{align}
and finally 
\begin{equation}
\begin{aligned}
\ket{\Psi_0}=&\mathrm{sign}(v_0)\sqrt{1-\left(\frac{v_0}{v_\infty}\right)^2}\ket{0}+\text{sign}(v_0)\left|\frac{v_0}{v_\infty}\right|\ket{1}\\
=&\text{sign}(v_0)\sqrt{1-\left(\frac{v_0}{v_\infty}\right)^2}\ket{0}+\frac{v_0}{v_\infty}\ket{1}.
\end{aligned}
\end{equation}

\begin{figure*}[tpbh]
    \centering
    \input{circuit_tex_files/controlled_rotations}
    \caption{The controlled $R_y(\phi_i), i \in \{0,\dots,L-1\}$ rotations that turn the uniform superposition of computational basis states into a weighted one. }
    \label{controlled_rotations}
\end{figure*}
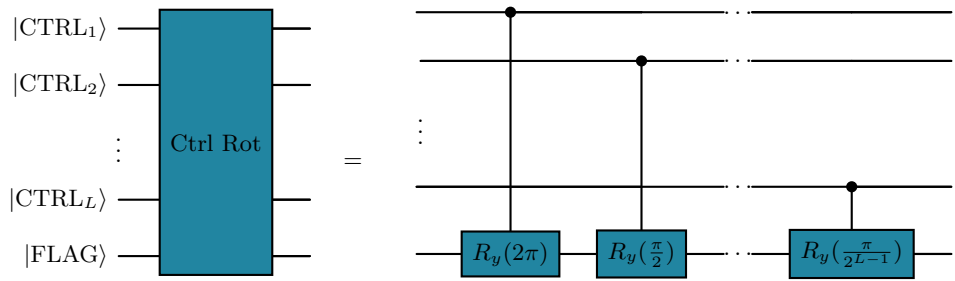
As described in Section~\ref{state_evolution}, these controlled rotations have the effect of turning the uniform superposition of computational basis states in Eq.~\ref{uniform} into the weighted superposition in Eq.~\ref{weighted_superposition}.

\subsection{The Amplitude Amplification Algorithm }\label{ampAmplification}
Looking at the decomposition in Eq.~\ref{weighted_superposition}, it comes natural to apply the amplitude amplification procedure, which is a generalization of the famous Grover's algorithm, to amplify the probability of measuring the desired state \cite{Brassard_2002,grover1996fast}.  
 We repeatedly apply the Grover operator
 \begin{equation}
\mathcal{Q}=-\mathcal{E}\mathcal{S}_0\mathcal{E}^\dagger\mathcal{S}
\end{equation}
to the FLAG and SYS registers, where $\mathcal{E}$ is the full encoding circuit described above seen as a gate while $\mathcal{S}$ and $\mathcal{S}_0$ are respectively the oracle and the reflection around the zero state.
More details about the procedure can be found in supplementary material~\ref{AppendixD}.

The original encoding algorithm prepares the right superposition in the final state only whenever we measure the value $1$ on the FLAG qubit. This happens with probability $\rho$. In case we measure $0$ instead we have to repeat the procedure from scratch. It follows that the average amount of runs we need for a successful encoding scales as $\mathcal{O}(\frac{1}{\rho})$. 
Using the amplitude amplification algorithm we boost the probability of measuring the correct value on the FLAG register, so that it reaches a value of at least $\max(1-\rho, \rho)$.
As usual when applying Grover's algorithm, this procedure gives us a quadratic improvement over the runtime, that now scales as $\mathcal{O}(\frac{1}{\sqrt\rho})$.

\section{Acknowledgements}
VP acknowledges helpful discussions with Andreas Buchheit, Peter Schuhmacher, and Christian Wimmer.
This study was funded by the QuantERA grant EQUIP via DFG project 491784278 and by the
Federal Ministry for Economics and Climate Action (BMWK) via project ALQU and the Quantum Fellowship Program of DLR. The funder played no role in study design, data collection, analysis and interpretation of data, or the writing of this manuscript.

\section{Data availability}
Data is available from the corresponding author upon reasonable request.
\section{\textbf{CODE AVAILABILITY}}
The codes used to generate data for this paper are available from the corresponding
author upon reasonable request.

\section{Author contributions}
ME and MF secured the funding for this project. VP came up with the project idea. VP designed the algorithm and proved its runtime with support of ME. VP carried out the numerical case study with support of SH. All authors contributed to the analysis and interpretation of the results and the writing of the manuscript.

\section{Competing interests}
All authors declare no financial or non-financial competing interests.

\FloatBarrier
%\section{References}
%\bibliography{sample}
\input{main.bbl}

%%%%%%%%%% Merge with supplemental materials %%%%%%%%%%
\clearpage
%\widetext
\begin{center}
\textbf{\large Supplementary Material: Fast Quantum Amplitude Encoding of Typical Classical Data}
\end{center}
%%%%%%%%%% Merge with supplemental materials %%%%%%%%%%
%%%%%%%%%% Prefix a "S" to all equations, figures, tables and reset the counter %%%%%%%%%%
\setcounter{section}{0}
\setcounter{equation}{0}
\setcounter{figure}{0}
\setcounter{table}{0}
\setcounter{page}{1}
\makeatletter
\renewcommand{\thesection}{S\arabic{section}}
\renewcommand{\theequation}{S\arabic{equation}}
\renewcommand{\thefigure}{S\arabic{figure}}
\renewcommand{\bibnumfmt}[1]{[S#1]}
\renewcommand{\citenumfont}[1]{S#1}
%%%%%%%%%% Prefix a "S" to all equations, figures, tables and reset the counter %%%%%%%%%%

\section{Behavior of the total runtime}\label{AppendixA}
In Section~\ref{timeCost} we showed that the total runtime for the encoding algorithm is 
\begin{equation}
\tau \sim n\frac{1}{\sqrt \rho}. \label{eq:tau}
\end{equation}
Since the vector $\vc{v}$ is normalized using the Euclidean norm as $\|\vc{v}\|_2=1$, we have
\begin{equation}
\rho(\vc{v})=\frac{1}{N}\frac{\|\vc{v}\|_2}{\|v\|_\infty}=\frac{1}{N \|\vc{v}\|_\infty^2}. \label{eq:rhoofv}
\end{equation}
Inserting Eq.~(\ref{eq:rhoofv}) into Eq.~(\ref{eq:tau}) yields
\begin{equation}
\tau \sim \sqrt{N \|\vc{v}\|_\infty^2} n.
\end{equation}

\subsection{Average runtime for random input vectors \label{averageRuntime}}
In this section we will calculate the average runtime when the input vectors are uniformly sampled at random from $\mathbf S_{N-1}$.

For $j\in\{1,2,...,N\}$ we define the $j$-th canonical basis vector with a sign $(-1)^s$, $s\in\{0,1\}$, whose $i$-th component is
\begin{equation}
    \vc{v}(j, s)_i:=(-1)^s \delta_{ij}. \label{eq:basiswithsign}
\end{equation}
We also define a uniform (in modulus) vector
\begin{equation}
 \vc v(\vc{s}) := \sum_j \vc{v}(j, s_j) \label{eq:uniformvector}
\end{equation}
for $\vc{s}\in\{0,1\}^N$.
Note that the vectors $\vc{v}(j, s)$ and $\vc{v}(\vc{s})$ are maximally sparse and maximally dense, respectively.
There are only $2N$ different $\vc{v}(j, s)$, but there are $2^N$ different $\vc{v}(\vc{s})$.

The squared max norm $\|\cdot \|_\infty^2$ restricted to the $N$-sphere is a continuous function to the (positive) interval $[\frac{1}{N}, 1]$
\begin{equation}
    \|\cdot \|_\infty^2:\mathbf{S}_{N-1}\to [\frac{1}{N}, 1].
\end{equation}
It reaches its maximum $1$ only on the $\vc{v}(j,s)$ and its minimum $\frac{1}{N}$ only on the $\vc{v}(\vc{s})$.
We immediately notice that there are many more vectors with value of $\|v\|_\infty^2$ equal to $\frac{1}{N}$ than to $1$.

This fact can be seen as the vectors with $\|v\|_\infty^2=N$ having a lower entropy (induced by the function) than the ones for which $\|v\|_\infty^2=1$.

The square of the infinity norm is invariant under a change of sign for any coordinate and also under any permutation of them.
Thus the $(N-1)$-sphere can be partitioned into $2^N\cdot N!$ regions associated with different choices $\vc{s}$ for the signs and permutations $\sigma$ of the components.
Each of these regions has the same area
\begin{equation}
D(N-1):=\frac{2 \pi^{\frac{N}{2}}}{\Gamma(\frac{N}{2})2^N N!}.
\end{equation}
We now compute the average vale of $\|v\|_\infty^2$ over a uniform distribution on $S_{N-1}$, which we denote by $\langle\|v\|_\infty^2 \rangle_{S_{N-1}}$. %Because of the symmetries described before, this quantity can be computed in the smallest non-periodic region of $S_{N-1}$, i.e. by fixing signs and an ordering of the coordinates.

A standard way to generate uniformly random vectors on the $(N-1)$-sphere is to generate a vector $\vc{x}$ with components $x_i$ sampled from independent Gaussian distributions and then normalize it.
The resulting distribution is uniform on the $(N-1)$-sphere, because the probability density for $\vc{x}$, $\frac{1}{(2\pi)^{N/2}} \mathrm{e}^{-\|\vc{x}\|_2^2}$, is rotationally invariant.
This connection between the normal distribution $\mathcal{N}_n(0,1)$ and the uniform distribution over the $(N-1)$-sphere $S_{N-1}$ implies that the average value of $\|v\|_\infty^2$ can be computed via
\begin{equation}
\langle\|v\|_\infty^2\rangle_{S_{N-1}}=\left\langle \frac{\| v\|_{\infty}^2}{\|v\|_2^2}\right\rangle_{\mathcal{N}_N(0,1)}.
\end{equation}
The expectation value of a ratio of two random variables can be obtained via a Taylor series around the mean values~\cite{stuart1998kendalls}. Up to second order in the distance from the mean it reads
\begin{equation}
\begin{aligned}
\left\langle\frac{\|x\|_{\infty}^2}{\|x\|_2^2}\right\rangle
\approx&\hphantom{+{}} \frac{\langle\| x\|_{\infty}^2\rangle}{\langle\|x\|_2^2\rangle}-\frac{\mathrm{Cov}(\| x\|_{\infty}^2,\| x\|_{2}^2)}{\langle\|x\|_2^2\rangle^2}+\\
&+\frac{\langle\| x\|_{\infty}^2\rangle \mathrm{Var}(\| x\|_{2}^2)}{\langle\|x\|_2^2\rangle^3}
\end{aligned}\label{expectedratio}
\end{equation}
and the first order expansion of the variance gives
\begin{align}
\mathrm{Var}\left(\frac{\|x\|_{\infty}^2}{\|x\|_2^2}\right)\approx&\hphantom{+{}}
\frac{\mathrm{Var}\left(\|x\|_{\infty}^2\right)}{\langle \|x\|_2^2 \rangle^2}\nonumber\\
&- 2 \frac{\langle \|x\|_\infty^2 \rangle}{\langle \|x\|_2^2 \rangle^3} \mathrm{Cov}(\| x\|_{\infty}^2,\| x\|_{2}^2)\nonumber\\
&+ \frac{\langle \|x\|_\infty^2 \rangle^2}{\langle \|x\|_2^2 \rangle^4} \mathrm{Var}\left(\|x\|_{2}^2\right).\label{varianceofratio}
\end{align}
We now continue to calculate all the quantities required to evaluate Eqs.~(\ref{expectedratio}) and (\ref{varianceofratio}).
The expectation value for the max norm is
\begin{equation}
    \langle \|x\|_\infty \rangle_{\mathcal{N}(0,1)} 
    %= \left(\sqrt{2\log(n)} + \mathcal{O}\left(\frac{1}{\log(n)}\right)\right)^2.
    = \left(\sqrt{2 \ln (N)} +\mathcal{O}(\frac{1}{\sqrt{\ln(N)}})\right),
\end{equation}
see Eq.~(28.6.13) in \cite{cramer1946mathematical}, where we will need the cases $k=2$ and $k=4$.
For the Euclidean norm the expectation value reads
\begin{equation}
    \langle\| x\|_{2}^k\rangle_{\mathcal{N}(0,1)} = 2^{\frac{k}{2}} \frac{\Gamma(\frac{N+k}{2})}{\Gamma(\frac{N}{2})},
\end{equation}
which is the $k$-th moment of the Chi distribution~\cite{abell1999statistics}.
Again we are interested in the special cases
\begin{equation}
\begin{aligned}
    \langle\| x\|_{2}^2\rangle_{\mathcal{N}(0,1)} =& N\\
\text{and }    \langle\| x\|_{2}^4\rangle_{\mathcal{N}(0,1)} =& (N+2) N.
\end{aligned}
\end{equation}
For what concerns the infinity norm we can prove the following
\begin{lemma}
The expectation value of the square of the infinity norm over a normal distribution with $\sigma=1$ behaves for large values of N as

\begin{equation}
\langle \| x\|_{\infty}^k\rangle _{\mathcal{N}(0,1)}\sim 2^\frac{k}{2} \ln(N)^\frac{k}{2}
\end{equation}
\end{lemma}
\begin{proof}
Let $Z_i:=x_i$ and $M:=M_N:=\|x\|_{\infty}=\max _1^N\left|Z_i\right|$.
\begin{equation}
E [M^k]=\int_0^{\infty} P\left(M^k>u\right) d u
\end{equation}
So, for real $u>0$ we need to study 
\begin{equation}
\begin{aligned}
P\left(M^k>u\right)=&P(M>u^\frac{1}{k})\\
=&1-P\left(\max _1^N\left|Z_i\right| \leq u^\frac{1}{k}\right)\\
=&1-P\left(\left|Z_1\right| \leq u^\frac{1}{k}\right)^N
\end{aligned}
\end{equation}
We define
\begin{equation}
G(x):=P\left(Z_1>x\right) =1-\Phi(x)=\frac{1}{2}\text{erfc}(\frac{x}{\sqrt 2}),
\end{equation}
where 
\begin{equation}
\begin{aligned}
    \Phi(x)=P(Z_1\leq x)=&\frac{1}{2\pi}\int_{-\infty}^{x}e^{-\frac{t^2}{2}}\\
    =&\frac{1}{2}[1+\text{erf} (\frac{x}{\sqrt 2})]
\end{aligned}
\end{equation}
and 
\begin{equation}
    \text{erf}(x)+\text{erfc}(x)=1
\end{equation}
we write 
\begin{align}
P\left(\left|Z_1\right| \leq u^\frac{1}{k}\right)&=P(-u^\frac{1}{k} \leq Z_1 \leq u^\frac{1}{k})\nonumber\\
&=P(Z_1\leq u^\frac{1}{k})-P(Z_1\leq -u^\frac{1}{k}).
\end{align}
In terms of the cumulative distribution function (CDF), this becomes \cite{Handbook_of_math}
\begin{align}
&P(Z_1\leq u^\frac{1}{k})-P(Z_1\leq -u^\frac{1}{k})\nonumber\\
=&\Phi(u^\frac{1}{k}) -\Phi(-u^\frac{1}{k})\nonumber\\
=&\Phi(u^\frac{1}{k}) - (1-\Phi(u^\frac{1}{k}) )\nonumber\\
=&2\Phi(u^\frac{1}{k})-1\nonumber\\
=& 2(1-G(u^\frac{1}{k}))-1 \nonumber\\
=&1-2G(u^\frac{1}{k})
\end{align}
where we used the symmetry 
\begin{equation}
    \Phi(-x)=1-\Phi(x) .
\end{equation}
Substituting into the original expression we get
\begin{equation}
    P(M^2>u)= 1-(1-2 G(u^\frac{1}{k}))^N.
\end{equation}
The Taylor expansion of $\Phi(x)$ 
around $x=\infty$ gives \cite{Handbook_of_math} 
\begin{equation}
   \Phi(x) = 1-  e^{-x^2 / 2}\left(\frac{1}{x \sqrt{2 \pi}} +\mathcal{O}\left(\frac{1}{x^2}\right)\right)
\end{equation}
so that 
\begin{equation}
   G(x)=1-\Phi(x)= e^{-x^2 / 2}\left(\frac{1}{x \sqrt{2 \pi}} +\mathcal{O}\left(\frac{1}{x^2}\right)\right).
\end{equation}
We define the new function 
\begin{equation}
g(u):=-\ln (1-2 G(u^\frac{1}{k})) 
\end{equation}
such that 
\begin{equation}
     1-2G(u^\frac{1}{k})=e^{-g(u)}
\end{equation}
so that we can write 
\begin{equation}
   P(M^2>u)= 1-(1-2 G(u^\frac{1}{k}))^N=1-e^{-N g(u)}  .
\end{equation}
Fig~\ref{erf} shows the behavior of $e^{-N g(u)}$ for $N=32$.\\
For large $u$ we have 
\begin{align}
g(u)=&-\ln (1-2 G(u^\frac{1}{k}))\overset{u \to \infty}{=}2 G(u^\frac{1}{k})+\mathcal{O}(G^2) \nonumber\\
=&\exp\left(- {\frac{u^\frac{2}{k}}{2}}\right)\left(\sqrt \frac{2}{\pi}\frac{1}{u^{\frac{1}{k}}}+\mathcal{O}\left(\frac{1}{u^{\frac{k+1}{k}}}\right)\right)
\end{align}
and therefore we have 
\begin{align}
P(M^2>u)=&1-e^{-N g(u)} = Ng(u) +\mathcal{O}(g^2)\nonumber\\
=&N\exp(- {\frac{u^\frac{2}{k}}{2}})(\sqrt \frac{2}{\pi}\frac{1}{u^{\frac{1}{k}}}+\mathcal{O}(\frac{1}{u^{\frac{k+1}{k}}})).
\end{align}

\begin{figure}
    \centering
    \includegraphics[width=\linewidth]{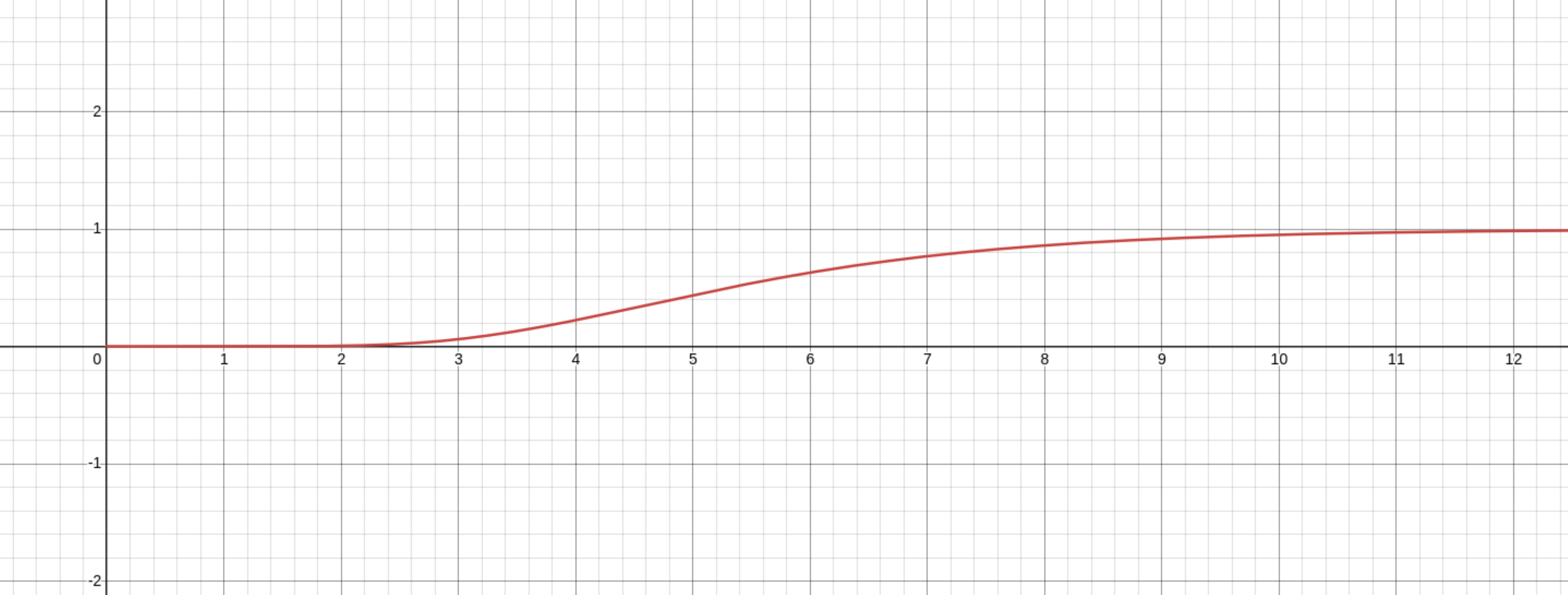}
    \caption{$y(u)=e^{-Ng(u)}=\text{erf}(\sqrt \frac{u}{ 2})^N$ for $N=32$}
    \label{erf}
\end{figure}

From the properties of the error function it follows that $g(u)$ decreases from $\infty$ to 0 as $u$ increases from 0 to $\infty$. 
Therefore, for each natural $N \geq 3$ there are unique positive real numbers $u_N$ and $v_N$ such that
\begin{equation}
g\left(u_N\right)=\frac{\ln N}{N}, \quad  g\left(v_N\right)=\frac{1}{N} .
\end{equation}
Clearly, $0<u_N<v_N<\infty$.
Moreover, for large N and therefore large $u_N $ and $v_N$ we can use $g(u)$ expansion to get the two equations 
\begin{align}
\frac{\ln N}{N}=&g\left(u_N\right)=\exp\left(- {\frac{u_N^\frac{2}{k}}{2}}\right)\left(\sqrt \frac{2}{\pi}\frac{1}{u_N^{\frac{1}{k}}}+\mathcal{O}\left(\frac{1}{u_N^{\frac{k+1}{k}}}\right)\right), \\
\intertext{and}
\frac{1}{N}=&g(v_N)=\exp\left(- {\frac{v_N^\frac{2}{k}}{2}}\right)\left(\sqrt \frac{2}{\pi}\frac{1}{v_N^{\frac{1}{k}}}+\mathcal{O}\left(\frac{1}{v_N^{\frac{k+1}{k}}}\right)\right). 
\end{align}
If we invert this expressions we get
\begin{equation}
e^t t \left[1+\mathcal{O}\left(\frac{f(N)}{e^\frac{t}{2}t^\frac{1+k}{2}}\right)\right]=\frac{2}{\pi }f(N)^2.
\end{equation}
where we define $t:= u^{\frac{2}{k}}$ and $f(N):=\frac{1}{g(u)}$.
Here we assume that $\frac{f(N)}{e^\frac{t}{2}t^\frac{1+k}{2}} \underset {N \to \infty}{\to 0}$.
We will later check if this assumption is consistent with the result.
Solving for the leading term, we have 
\begin{alignat}{2}
    &&e^t t =&\frac{2}{\pi }f(N)^2\\
    \Rightarrow&& t=&\mathcal{W}(\frac{2}{\pi }f(N)^2)
\end{alignat}
which implies 
\begin{equation}\label{A33}
 u(n)=\mathcal{W}\left(\frac{2}{\pi }f(N)^2\right)^{\frac{k}{2}},
\end{equation}
where $\mathcal{W}$ is the Lambert function. 
This well studied function has an asymptotic scaling given by \cite{corless1996}
\begin{equation}
    \mathcal{W}(x)=\ln(x)-\ln\ln(x)+o(1),
\end{equation}
which we can apply to get the approximation
\begin{align*}
u(N)&=[\ln(\frac{2}{\pi }f(N)^2)-\ln\ln(\frac{2}{\pi }f(N)^2)+o(1)]^\frac{k}{2}=\\
&=[2\ln(f(N))-\ln\ln(f(N))+\mathcal{O}(1)]^\frac{k}{2}.
\end{align*}
Using Newton multinomial formula \cite{Stanley2012} we find
\begin{equation}
\begin{aligned}
    u(N)=&2^\frac{k}{2}\ln^\frac{k}{2}(f(N))+\\
    &+
    \mathcal{O}(k2^{\frac{k}{2}-1}\ln^{\frac{k}{2}-1}(f(N))\ln\ln(f(N)),
\end{aligned}    
\end{equation}
which implies
\begin{align}
    u_N=&2^\frac{k}{2}\ln^\frac{k}{2}(N)+\mathcal{O}(k2^{\frac{k}{2}-1}\ln^{\frac{k}{2}-1}(N)\ln\ln(N))\\
    \intertext{and}
     v_N=&2^\frac{k}{2}\ln^\frac{k}{2}(N)+\mathcal{O}(k2^{\frac{k}{2}-1}\ln^{\frac{k}{2}-1}(N)\ln\ln(N)).
\end{align}
Next we consider the following expectation value
\begin{equation}
\begin{aligned}
E [M^2]=&\int_0^{\infty} P\left(M^2>u\right) du\\
=&\int_0^{\infty}\left(1-e^{-N g(u)}\right) du\\
=&I_1+I_2+I_3,
\end{aligned}
\end{equation}
where
\begin{align}
I_1:=&\int_0^{u_N}\left(1-e^{-N g(u)}\right) du,\\
I_2:=&\int_{u_N}^{v_N}\left(1-e^{-N g(u)}\right) du\\
\text{and }I_3:=&\int_{v_{\mathrm{N}}}^{\infty}\left(1-e^{-N g(u)}\right)du.
\end{align}
We will now focus on these three integrals one by one.
\subsubsection{Integral \texorpdfstring{$I_1$}{I1}}

If $0<u<u_N$, then $0<e^{-N g(u)}<e^{-N g\left(u_N\right)}=1 / N$ because, as shown in Fig.~\ref{erf}, this function is strictly increasing in that interval.
So,
\begin{align*}
u_n\geq I_1 =&
u_N-\int_0^{u_N}\left(e^{-N g(u)}\right) d u\\
\geq& u_N-\int_0^{u_N} \frac{1}{N}d u\\
=&u_n(1-\frac{1}{N})
\end{align*}
so that for large $N$
\begin{equation}
    I_1\sim  u_N =2^\frac{k}{2}\ln^\frac{k}{2}(N)+\mathcal{O}\left(k2^{\frac{k}{2}-1}\ln^{\frac{k}{2}-1}(N)\ln\ln(N) \right)\\
\end{equation}
\subsubsection {Integral \texorpdfstring{$I_2$}{I2}}
For the second integral we start by noticing that
\begin{equation}
\frac{1}{N}=e^{-N g(u_N)}\leq e^{-N g(u)}\leq e^{-N g(v_N)}=\frac{1}{e}
\end{equation}
from which we derive 
\begin{equation}
(v_N-u_N)(1-\frac{1}{e})\leq I_2 \leq (v_N-u_N)(1-\frac{1}{N})
\end{equation}
so that 
\begin{equation}
    I_2 \underset{N \to \infty}{\sim} v_N-u_N =\mathcal{O}(k2^{\frac{k}{2}-1}\ln^{\frac{k}{2}-1}(N)\ln\ln(N))
\end{equation}
\subsubsection {Integral \texorpdfstring{$I_3$}{I3}}
Using Eq.~(\ref{A33}) for large N and u (justified by the fact that $v_n \underset {N \to \infty} {\sim} 2\ln(N)$) we have 
\begin{align}
e^{-Ng(u)}\underset{u \to \infty} {\sim}& (1-e^{\frac{u}{2}}[\sqrt{\frac{2}{ \pi u }}+o(\frac{1}{u})])^N\nonumber\\
\underset{u,N \to \infty} {\sim}& 1-Ng(u)\nonumber\\
=&1-Ne^{\frac{-u}{2}}\sqrt{\frac{2}{ \pi u }}
\end{align}
\begin{align*}
I_3&\leq \int_{v_N}^{\infty} N g(u) d u \sim \int_{v_{\mathrm{N}}}^{\infty} N \frac{2}{\sqrt{2 \pi u}} e^{-u / 2} d u \\
&=2 N(1- \text{erf}(\sqrt{\frac{v_N}{2}}))=2Ng(v_N)=2
\end{align*}
We conclude that
\begin{align}
\langle \|x\|_{\infty}^k\rangle_{\mathcal{N}(0,1)}=&E [M^k]\nonumber\\
 =&I_1+I_2+I_3\nonumber\\
 =&2^\frac{k}{2}\ln^\frac{k}{2}(N)+\mathcal{O}\left(k 2^{\frac{k}{2}-1}\ln^{\frac{k}{2}-1}(N)\ln\ln(N)\right).
\end{align}
\end{proof}

In order to evaluate to evaluate Eqs.~(\ref{expectedratio}) and (\ref{varianceofratio}), we need the following results:
\begin{align}
\langle[\| x\|_{\infty}^2\rangle_{\mathcal{N}(0,1)}=& 2 \ln(N)+\mathcal{O}(\ln\ln(N)),\\
\langle\| x\|_{\infty}^4\rangle_{\mathcal{N}(0,1)}=& 4\ln^2(N)+\mathcal{O}(\ln(N)\ln\ln(N)),\\
\text{Var}(\|x\|_2^2)_{\mathcal{N}(0,1)}=&\langle\|x\|_2^4\rangle_{\mathcal{N}(0,1)}-\langle\|x\|_2^2\rangle_{\mathcal{N}(0,1)}^2\nonumber\\
=&N(N+2)-N^2=2N,\\
\text{Var}(\|x\|_{\infty}^2)_{\mathcal{N}(0,1)}=&\langle\|x\|_\infty^4\rangle_{\mathcal{N}(0,1)}-\langle\|x\|_\infty^2\rangle_{\mathcal{N}(0,1)}^2\nonumber\\
=&4\ln^2(N)+o(\ln^2(N))-4\ln^2(N)\nonumber\\
=&\mathcal{O}(\ln(N)\ln\ln(N)),
\end{align}
and 
\begin{align}
&\text{Cov}(\|x\|_\infty^2,\|x\|_2^2)_{\mathcal{N}(0,1)}\nonumber\\
=&\langle\|x\|_2^2\cdot\|x\|_\infty^{2}\rangle_{\mathcal{N}(0,1)}-\langle\|x\|_2^2\rangle_{\mathcal{N}(0,1)} \langle \|x\|_\infty^{2}\rangle_{\mathcal{N}(0,1)}.
\end{align}
Here we can use the covariance inequality
\begin{align}
\text{Cov}(\|x\|_\infty^2,n\|x\|_2^2)\leq&\sqrt{Var(\|x\|_\infty^2)Var(\|x\|_2^2)}\nonumber\\
=&\sqrt{(2N)o(\ln^2())}.
\end{align}
Substituting all of our results back into Eq.~(\ref{expectedratio}), for large N we have
\begin{align}
\langle \frac{\| x\|_{\infty}^2}{\|x\|_2^2}\rangle_{\mathcal{N}(0,1)}=&\frac{2\ln(N)}{N}(1-\frac{\sqrt{(2N)o(\ln^2(N))}}{2N\ln(N)}+\frac{2N} {N^2})\nonumber\\
=&\frac{2\ln(N)}{N}(1-o(1)).
\end{align}
\subsubsection{Variance}
For the variance we have the following
\begin{lemma}
\begin{equation}
\text{Var}(\frac{\| x\|_{\infty}^2}{\|x\|_2^2})_{\mathcal{N}(0,1)}= \mathcal{O}(\frac{\ln(N)}{N^2})
\end{equation}

\begin{proof}
we first need to compute the value of 
\begin{equation}
\langle\frac{\| x\|_{\infty}^4}{\|x\|_2^4}\rangle_{\mathcal{N}(0,1)}
\end{equation}
we can again use the expansion 
\begin{align}
\langle \frac{\| x\|_{\infty}^4}{\|x\|_2^4}\rangle\approx& \frac{\langle\| x\|_{\infty}^4\rangle}{\langle\|x\|_2^4\rangle}(1-\frac{\text{Cov}(\| x\|_{\infty}^4,\| x\|_{2}^4)}{\langle\| x\|_{\infty}^4\rangle\cdot \langle\| x\|_{2}^4\rangle}+\frac{Var(\| x\|_{2}^4)}{(\langle\| x\|_{2}^4\rangle)^2})\nonumber\\
=&\hphantom{+{}}\frac{4\ln^2(N)+\mathcal{O}(\ln(N))}{(N+2)N}(1+\nonumber\\
&-\frac{\text{Cov}(\| x\|_{\infty}^4,\| x\|_{2}^4)}{(4\ln^2(N)+\mathcal{O}(\ln^2(N)))\cdot ((N+2)N)}+\nonumber\\
&+\frac{Var(\| x\|_{2}^4)}{(N(N+2))^2})
\end{align}
We look at the value of 
\begin{align}
\text{Var}(\| x\|_{2}^4)_{\mathcal{N}(0,1)}&=\langle\| x\|_{2}^8\rangle_{\mathcal{N}(0,1)}-\langle\| x\|_{2}^4\rangle_{\mathcal{N}(0,1)}^2\nonumber\\
&=\mathcal{O}(N^3)
\end{align}
while
\begin{align}
\text{Var}(\| x\|_{\infty}^4)_{\mathcal{N}(0,1)}&=\langle\| x\|_{\infty}^8\rangle_{\mathcal{N}(0,1)}-\langle\| x\|_{\infty}^4\rangle_{\mathcal{N}(0,1)}^2\nonumber\\
&=16 \ln^4(N)-(4 \ln^2(N)+\mathcal{O}(\ln(N)))^2\nonumber\\
&=\mathcal{O}(\ln^3(N))
\end{align}
Finally for the covariance we have the inequality  \cite{mukhopadhyay2020probability}
\begin{align}
&\text{Cov}(\|x\|_\infty^4,\|x\|_2^4)_{\mathcal{N}(0,1)}\nonumber\\
\leq&
\sqrt{\text{Var}(\|x\|_\infty^4)_{\mathcal{N}(0,1)}\text{Var}(\|x\|_2^4)_{\mathcal{N}(0,1)}}
\nonumber\\
=&\mathcal{O}(N^\frac{3}{2}\ln^\frac{3}{2}(N))
\end{align}
therefore 
\begin{equation}
\frac{\text{Cov}(\|x\|_\infty^4,\|x\|_2^4)_{\mathcal{N}(0,1)}}{\langle\|x\|_2^4\rangle_{\mathcal{N}(0,1)} \langle\|x\|_\infty^{4}\rangle_{\mathcal{N}(0,1)}}\leq \frac{\mathcal{O}(N^\frac{3}{2}\ln^\frac{3}{2}(N))}{\mathcal{O}(N^2\ln^2(N))}
=o(1)
\end{equation}
For the remaining term in the approximation we have 
\begin{equation}
\frac{\text{Var}(\| x\|_{2}^4)_{\mathcal{N}(0,1)}}{\langle \| x\|_{2}^4\rangle_{\mathcal{N}(0,1)}^2}=\frac{4
N(N+2)(2N+3)}{(N+2)^2N^2}=o(1)
\end{equation}
So we can state that in the limit of large N 
\begin{equation}
\left\langle\frac{\| x\|_{\infty}^4}{\|x\|_2^4}\right\rangle_{\mathcal{N}(0,1)}=  \frac{4\ln^2(N))}{N^2}+\mathcal{O}(\frac{\ln(N)}{N^2}).
\end{equation}
We already know from before that $\left\langle \frac{\| x\|_{\infty}^2}{\|x\|_2^2}\right\rangle\sim \frac{4\ln^2(N))}{N^2}$, hence 
\begin{align}
&\text{Var}(\frac{\| x\|_{\infty}^2}{\|x\|_2^2})_{\mathcal{N}(0,1)}\nonumber\\
=&\langle \frac{\| x\|_{\infty}^4}{\|x\|_2^4}\rangle_{\mathcal{N}(0,1)} -\langle \frac{\| x\|_{\infty}^2}{\|x\|_2^2}\rangle_{\mathcal{N}(0,1)} ^2 \nonumber\\
=&\mathcal{O}(\frac{\ln(N)}{N^2}).
\end{align}
\end{proof}
\end{lemma}
\subsubsection{Conclusions}
Summarizing the results we have
\begin{align}
\langle \| x\|_{\infty}^2\rangle_{\mathcal{S}_{N-1}} =&\langle \frac{\| x\|_{\infty}^2}{\|x\|_2^2}\rangle_{\mathcal{N}(0,1)}=\frac{2\ln(N)}{N}(1-o(1))\\
\text{Var}(\| x\|_{\infty}^2)_{\mathcal{S}_{N-1}}
=& (\text{Var}(\frac{\| x\|_{\infty}^2}{\|x\|_2^2})_{\mathcal{N}(0,1)}\nonumber\\
=&\langle \frac{\| x\|_{\infty}^4}{\|x\|_2^4}\rangle_{\mathcal{N}(0,1)}-\langle \frac{\| x\|_{\infty}^2}{\|x\|_2^2}\rangle_{\mathcal{N}(0,1)}^2 \nonumber\\
=& \mathcal{O}(\frac{\ln(N)}{N^2}),
\end{align}
which in terms of $\tau$ translates to 
\begin{equation}
E[\tau(\vc v)]\sim 2\ln(N)
\end{equation}
\begin{equation}
Var(\tau(\vc v))=\mathcal{O}(\ln(N))
\end{equation}

\section{Sparse Redundant input vectors}

If the input vector is sparse or has many repeated values, the complexity of the encoding circuit decreases. 
For example, for a sparse vector such as 
\begin{equation}
\vc{v}=\frac{1}{\sqrt {a^2+b^2}}(a,0,0,0,0,0,0,b)
\end{equation}
we only have three possible values for the entries, namely $\frac{a}{\sqrt {a^2+b^2}} $,$\frac{b}{\sqrt {a^2+b^2}} $ and $0$.
The initial state for the ancilla, $\ket{0}$ already corresponds to the value 0, hence we only need as many Index registers as the amount of non zero values in the input vector.
This reasoning can also be extended to redundant vectors: we substitute the value 0 with the most present value in the vector and we then  shift accordingly all the entries in the angle vector $\vc \theta$.

We can conclude that the memory cost of the algorithm is directly proportional to the amount of non zero (or more generally of different) entries in the input vector. Since, as explained before, we can always shift $\vc \theta $ by its most frequent value, the memory scaling of the circuit that encodes an input vector of size $N=2^n$ with $S$ different values is $\mathcal{O}(Sn)$.

\section{State Evolution of the encoding circuit}
We can now show in detail the evolution of the initial state under the action of the circuit with M=N. This is the most parallel version of the circuit and therefore the fastest and, at the same time, the most memory consuming one.

In this description we make use of the intermediate vector $\vc c$ defined as
\begin{equation}
    c_i:=\frac{w_i}{\|\vc w \|_\infty}
\end{equation}
The initial state is:
\begin{equation}
   \ket{\Psi_0}=\ket{0}_{SYS}\bigotimes_{j=0}^{N-1}[\ket{1^{\otimes n}}_{I_j}\ket{0}_{C_j}]\ket{0^{\otimes L}}_{CTRL}\ket{0}_{FLAG}
\end{equation} 
The first step consists of applying the Hadamard gate 
$H^{\otimes n}$ to the system register to create a superposition of all the index values, together with setting the parity check register to $\ket{1^{\otimes n}}$ using n X gates on it.
As stated before, we can generalize this step to uniform superpositions of N states where N is not a power of two using the procedure introduced in \cite{shukla2023efficient}
\begin{align}
\ket{\Psi_1} &= \frac{1}{\sqrt N} \sum_{k=0}^{N-1} \ket{k}_{SYS} \bigotimes_{j=0}^{N-1} [\ket{1^{\otimes n}}_{I_j} \ket{0}_{C_j}] \nonumber\\
&\quad \ket{0^{\otimes L}}_{CTRL} \ket{0}_{FLAG}
\end{align}
In the XOR operation described in eq.\ref{Xor_gate}, the system register is entangled with the parity check register using a series of n C-NOTs which use the system qubits as control and one of the Index registers as targets.
\begin{align}
\ket{\Psi_2} =&\text{XOR}\ket{\Psi_1}\nonumber\\
=& \frac{1}{\sqrt N} \sum_{k=0}^{N-1} \ket{k}_{SYS} \bigotimes_{j=0}^{N-1} [\ket{1^{\otimes n} \oplus k}_{I_j} \ket{0}_{C_j}] \nonumber\\
&\otimes \ket{0^{\otimes L}}_{CTRL} \ket{0}_{FLAG}
\end{align}
We now load into each Index register the corresponding binary value, after applying the LoadIndex gates to the state $\ket{\Psi_2}$ we get
\begin{align}
\ket{\Psi_3} &= \frac{1}{\sqrt N} \sum_{k=0}^{N-1} \ket{k}_{SYS} \bigotimes_{j=0}^{N-1} [\ket{1^{\otimes n} \oplus k \oplus j}_{I_j} \ket{0}_{C_j}]\nonumber \\
&\quad \ket{0^{\otimes L}}_{CTRL} \ket{0}_{FLAG}
\end{align}
The next step is to perform the logical AND by means of a n-Toffoli gate between each Index register and the correspondent Parity compression registers, which are the target qubits. We notice that the state $\ket{1^{\otimes n}\oplus k \oplus j}$  is equal to $\ket{1^{\otimes n}}$ only for $k=j$ \\
\begin{align}
\ket{\Psi_4} =& \frac{1}{\sqrt N} \sum_{k=0}^{N-1} \left\{ \ket{k}_{SYS} \bigotimes_{j=0}^{N-1} [\ket{1^{\otimes n} \oplus k \oplus j}_{I_j} \ket{\delta_{k,j}}_{C_j}] \right\}\nonumber \\
&\ket{0^{\otimes L}}_{CTRL} \ket{0}_{FLAG}
\end{align}
The fifth step consists of the Fan-in operation that uses a series of classically controlled C-NOTs that have the parity compression registers as controls and the CTRL register as target.
This operation associates to each index $k$ loaded into the SYS register the correspondent row of the binary matrix $l$.
\begin{align}
\ket{\Psi_5}=&\frac{1}{\sqrt N}\sum_{k=0}^{N-1}\left\{\ket{k}_{SYS}\bigotimes_{j=0}^{N-1}[\ket{1^{\otimes n}\oplus k \oplus j}_{I_j}\ket{\delta_{k,j}}_{C_j}]\right\}\nonumber\\
&\bigotimes_{r=0}^{L-1}\ket{l_{k,r}}_{CTRL}\ket{0}_{FLAG}
\end{align}
The last relevant step of the encoding are now the controlled rotations, that turn the uniform superposition of basis states into a weighted one, where the weights are the entries of $\vc{c}$.
\begin{align}
\ket{\Psi_6}=&\frac{1}{\sqrt N}\sum_{k=0}^{N-1} \{ \ket{k}_{\mathrm{SYS}}\bigotimes_{j=0}^{N-1}[\ket{1^{\otimes n}\oplus k \oplus j}_{I_j}\ket{\delta_{k,j}}_{C_j}]\nonumber\\
&\bigotimes_{r=0}^{L-1}\ket{l_{k,r}}_{\mathrm{CTRL}}((-1)^{l_{k,0}} \sqrt{1 - c_k^2}\ket{0}+c_k\ket{1})_{\mathrm{FLAG}}\}
\end{align}
Which can be written as
\begin{equation}
  \ket{\Psi_7}=\ket{\Psi_{(0)}}\ket{0}_{\mathrm{FLAG}}+\ket{\Psi_{(1)}}\ket{1}_{\mathrm{FLAG}},
\end{equation}
where
\begin{equation}
\begin{aligned}
 \ket{\Psi_{(0)}}=&\frac{1}{\sqrt {N }}\sum_{k=0}^{N-1}\{(-1)^{l_{k,0}} \sqrt{1 - c_k^2}\ket{k}_{SYS}\\
 &\bigotimes_{j=0}^{N-1}[\ket{1^{\otimes n}\oplus k \oplus j}_{I_j}\ket{\delta_{k,j}}_{C_j}]\bigotimes_{r=0}^{L-1}\ket{l_{k,r}}_{CTRL}\}
\end{aligned}
\end{equation}
and
\begin{equation}
\begin{aligned}
\ket{\Psi_{(1)}}=&\frac{1}{\sqrt {N }}\sum_{k=0}^{N-1}\{c_k\ket{k}_{SYS}\nonumber\\
&\bigotimes_{j=0}^{N-1}[\ket{1^{\otimes n}\oplus k \oplus j}_{I_j}\ket{\delta_{k,j}}_{C_j}]\bigotimes_{r=0}^{L-1}\ket{l_{k,r}}_{CTRL}\}.
\end{aligned}
\end{equation}
Taking out the normalization factor we notice that
\begin{align}\ket{\Psi_{(0)}}=&\frac{\sqrt{Nv_\text{max}^2-1}}{\sqrt {Nv_\text{max}^2}}\frac{1}{\sqrt {Nv_\text{max}^2-1}}\nonumber \\
&\sum_{k=0}^{N-1}(-1)^{l_{k,0}} \sqrt{v_\text{max}^2 - v_k^2}\ket{k}_{SYS}\nonumber\\
&\bigotimes_{j=0}^{N-1}[\ket{1^{\otimes n}\oplus k \oplus j}_{I_j}\ket{\delta_{k,j}}_{C_j}]\bigotimes_{r=0}^{L-1}\ket{l_{k,r}}_{CTRL}
\end{align}
and 
\begin{align}
 \ket{\Psi_{(1)}}=&\frac{1}{\sqrt {Nv_\text{max}^2}}\sum_{k=0}^{N-1}\{v_k\ket{k}_{SYS}\nonumber \\
&\bigotimes_{j=0}^{N-1}[\ket{1^{\otimes n}\oplus k \oplus j}_{I_j}\ket{\delta_{k,j}}_{C_j}]\bigotimes_{r=0}^{L-1}\ket{l_{k,r}}_{CTRL}\},
\end{align}
where we used that $c_k=\frac{v_k}{v_\text{max}}$ and $\sum_{k=0}^{N-1}v_k^2=1$.
Finally we can write the state as 
\begin{equation}
  \ket{\Psi_{7}}=\sqrt{1-\rho}\ket{\Psi_{7,B}}\ket{0}_{FLAG}+\sqrt{\rho}\ket{\Psi_{7,G}}\ket{1}_{FLAG},
\end{equation}
where 
\begin{align}
  \ket{\Psi_{7,B}}=&\frac{1}{\sqrt {1-\rho}}\ket{\Psi_{(0)}}\\
\text{and }\ket{\Psi_{7,G}}=&\frac{1}{\sqrt {\rho}}\ket{\Psi_{(1)}}
\end{align}
The following operations play the role of disentangling the SYS and the FLAG registers from all the others.
We will analyze their effect on the $\ket{\Psi_{B}}$ and $\ket{\Psi_{G}}$ states.

The inverse of the Fan-in simply restores the value of the CTRL register to all zeros, the inverses of the AND, of the Index loader, of the XOR and a final X line perform the same operation on the Parity compression and Index registers, so that we eventually have the states
\begin{align}
\ket{\Psi_{8,B}}=&\frac{1}{\sqrt {Nv_\text{max}^2-1}}\sum_{k=0}^{N-1}(-1)^{l_{k,0}} \sqrt{v_\text{max}^2 - v_k^2}\ket{k}_{SYS}\nonumber\\
&\bigotimes_{j=0}^{N-1}[\ket{0^{\otimes n}}_{I_j}\ket{0}_{C_j}]\ket{0^{\otimes L}}_{CTRL}
\end{align}
and 
\begin{align}
\ket{\Psi_{8,G}}=&\sum_{k=0}^{N-1}v_k\ket{k}_{SYS}\nonumber\\
&\bigotimes_{j=0}^{N-1}[\ket{0^{\otimes n}}_{I_j}\ket{0}_{C_j}]\ket{0^{\otimes L}}_{CTRL}.
\end{align}
In case of $M\neq N$, we would now be ready to encode a second chunk of $M$ entries of $\vc{w}$. 
Since here we are describing the most parallel version of the algorithm, all of the entries are encoded at once and we can trace out all the ancilla registers, except for the FLAG one, and obtain the state
\begin{equation}
  \ket{\Psi_{9}}=\sqrt{1-\rho}\ket{\Psi_{B}}_{SYS}\ket{0}_{FLAG}+\sqrt{\rho}\ket{\Psi_{G}}_{SYS}\ket{1}_{FLAG}
  \label{eq35}
\end{equation}
with
\begin{align}
\ket{\Psi_{B}}=&\frac{1}{\sqrt {Nv_\text{max}^2-1}}\sum_{k=0}^{N-1}(-1)^{l_{k,0}} \sqrt{v_\text{max}^2 - v_k^2}\ket{k}_{SYS}\\
\intertext{and }\ket{\Psi_{G}}=&\sum_{k=0}^{N-1}v_k\ket{k}_{SYS}.
\end{align}

\section{The amplitude amplification algorithm}\label{AppendixD}
In this supplementary material we provide the readers with a more detailed explanation of this specific application of the well known amplitude amplification algorithm \cite{Brassard_2002}. 
As usual, we repeatedly apply the Grover operator
 \begin{equation}
\mathcal{Q}=-\mathcal{E}\mathcal{S}_0\mathcal{E}^\dagger\mathcal{S}
\end{equation}
to the FLAG and SYS registers.
Here $\mathcal{S}$, which is usually referred to as the Oracle, flips the sign of the states we wish to amplify. In our case these are the ones with the FLAG register in the state $\ket{1}$, s.t. the action of the oracle is
\begin{equation}
 \begin{aligned}
     \mathcal{S}\ket{\Psi_G}=&-\ket{\Psi_G}
     \mathcal{S}\ket{\Psi_B}=&\ket{\Psi_B}
 \end{aligned}
\end{equation}

 In general, finding an efficient oracle representation is not a trivial task, but in this specific case it is just the $Z$ gate acting on the FLAG qubit.
 In the following we will indicate with F and S respectively the FLAG and SYS registers.
   \begin{equation}
     \mathcal{S}=Z_{F}\otimes I_{S}
 \end{equation}
 \begin{figure}[!tbph]%
    \centering%
    \includegraphics[width=0.9\linewidth]{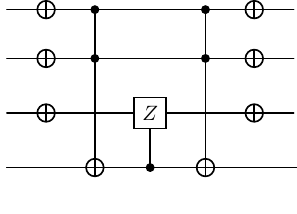}%
    \caption{This is the $S_0$ circuit for n=2. The controlled Z gate needs to act on the one among the n+1 non ancilla qubits which was not included in the n-Toffoli gate, but this choice doesn't change the effect of the gate. }%
    \label{S0}%
    \end{figure}%
 The gate $\mathcal{S}_0$ flips the sign of the all zero state $\ket{0}^{\otimes n+1}$, i.e.
   \begin{equation}
   \begin{aligned}
     \mathcal{S}_0\ket{x}=&-\ket{x},   &x=0^{n+1}\\
     \mathcal{S}_0\ket{x}=&\ket{x},   &x\neq0^{n+1},
     \end{aligned}
 \end{equation}
 which is equivalent to
    \begin{equation}
   \begin{split}
      \mathcal{S}_{0_{S,F}}=X_{S,F}(C^{0,..n-1}_{S}-Z_F)X_{S,F},
     \end{split}
 \end{equation}
 where $X_{S,F}$ means $X^{\otimes n+1}_{S,F}$, as in Fig.~\ref{S0}.
 Using an ancilla register A, we can write it as
\begin{align}   
\mathcal{S}_{0_{S,F,A}}
=&(X_{S,F}\otimes I_{A} )(C_{S}  I_{F}  X_A)(I_{S}\otimes Z_{F}C_{A} )\nonumber\\
&(C_{S}  I_{F}  X_A)(X_{F,S}\otimes I_{A} )
 \end{align}
, as shown in Fig~\ref{S0}.
The remaining $\mathcal{E}$ is our full encoding circuit seen as a gate.

This method is a generalization of Grover's algorithm, to which it reduces when we choose $\mathcal{E}$ to be $H^{\otimes(n+1)}$.
After m iterations the state of the system is
\begin{equation}
    Q^m \ket{\Psi}=sin[(2m+1)\theta)]\ket{\Psi_G}+cos[(2m+1)\theta)]\ket{\Phi_B}
\end{equation}
where $\sin(\theta) = \sqrt{\rho}$.
so the probability of measuring the desired state $\ket{\Psi_{G}}$ will be 
 \begin{equation}
 P^{m}_{G}=sin^2((2m+1)\theta))
\end{equation}
 The ideal number of iterations is therefore
 \begin{equation}
     m= \left\lfloor\frac{\pi}{4\theta}\right\rfloor=\left\lfloor\frac{\pi}{4 \arcsin(\sqrt{\rho})}\right\rfloor
 \end{equation} such that a good state (a basis state in the superposition $\ket{\Psi_{G}}$) is measured with a probability of at least $max(1-\rho, \rho)$.

 A visual interpretation of this algorithm, see Fig~\ref{amplitudeampgeom}, is similar to the one used for Grover's, consisting in a series of pairs of reflections of the state of the system, one with respect to the unwanted or bad state $\ket{\Psi_B}$ (orthogonal to $\ket{\Psi_{G}}$), which is performed by $S=(2\ket{\Psi}_B\bra{\Psi}_B-I)$, and the other one with respect to the initial state $\mathbf{\Psi}$ performed by $\mathcal{E}\mathcal{S}_0\mathcal{E^\dagger}=2\ket{\Psi}\bra{\Psi}-I.$ 
\begin{figure}[tbph]
    \centering
    \includegraphics[width=\linewidth]{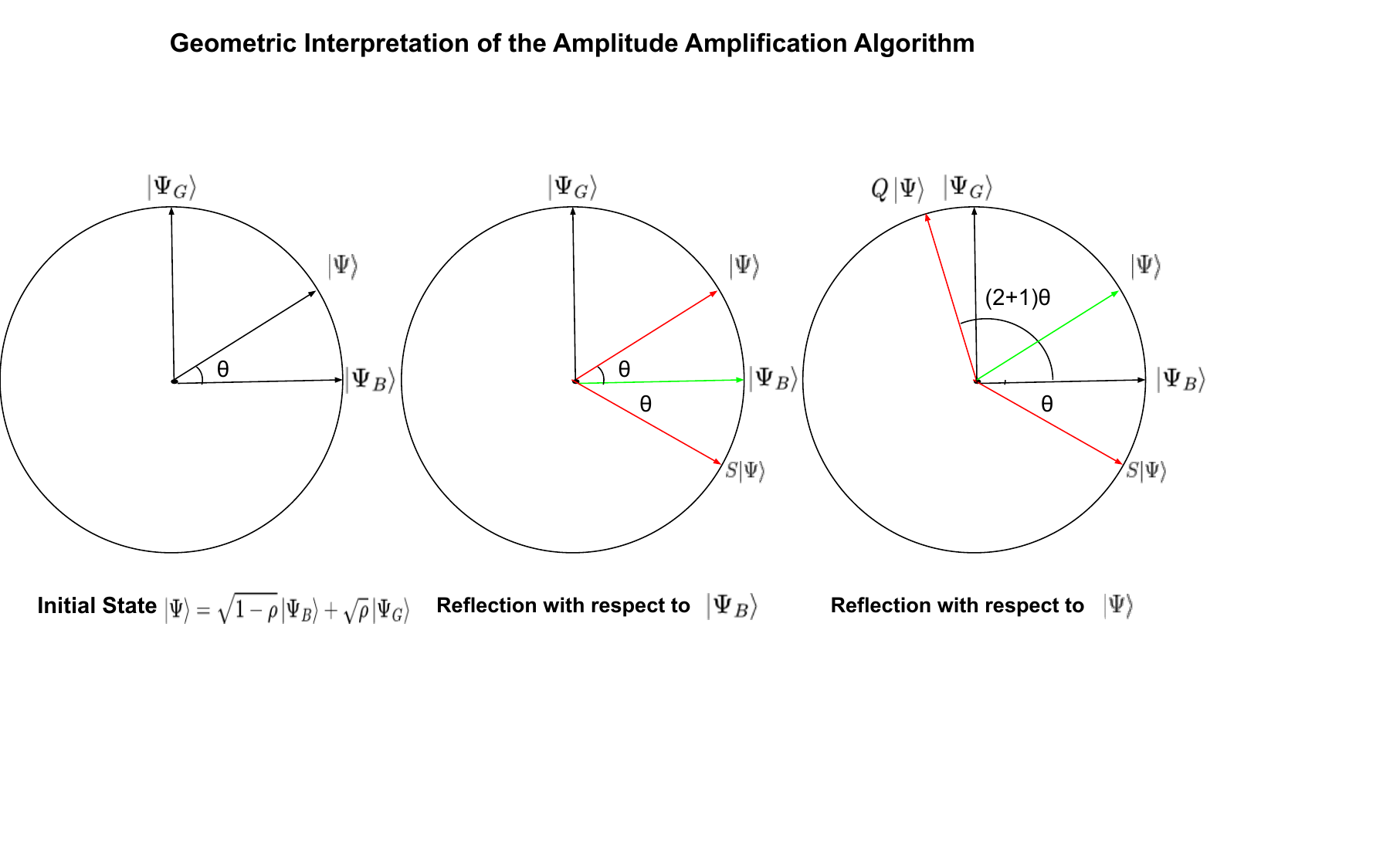}
    \caption{Geometric interpretation of the amplitude amplification algorithm.}
    \label{amplitudeampgeom}
\end{figure}

\end{document}

%% file: colors.tex
\definecolor{gate_color_H}{rgb}{0.92,0.49,0.43} 
\definecolor{gate_color_X}{rgb}{1,0.72,0.22} 
\definecolor{gate_color_xor}{rgb}{0.94,0.86,0.44}  
\definecolor{gate_color_im}{rgb}{0.751, 0.867, 0.79}
\definecolor{gate_color_and}{rgb}{0.67,0.87,0.94}
\definecolor{gate_color_mtc}{rgb}{0.32,0.65,0.75}    
\definecolor{gate_color_crot}{rgb}{0.13,0.52,0.63}
\definecolor{gate_group_color_a}{rgb}{0.8, 0.8, 0.8}
\definecolor{gate_group_color_b}{rgb}{0.9, 0.9, 0.9}

%% file: circuit_tex_files/encoding_circuit_latex.tex
\begin{quantikz}[row sep={0.8cm,between origins}, column sep=0.15cm]
    \lstick{$\text{SYS} \ket{0}$} &\gategroup[8,steps=3,style={dashed,rounded
corners,fill=gate_group_color_a, inner xsep=2pt},background,
label style={label position=above,anchor=north,yshift=0.4cm}]{
{\sc Initial Step}} &
\gate[style={fill=gate_color_H}]{H} & \gate[6,style={fill=gate_color_xor}][1cm]{\text{XOR}} &\qw & \gategroup[8,steps=8,style={dashed,rounded
corners,fill=gate_group_color_b, inner
xsep=2pt},background,label style={label
position=above,anchor=north,yshift=0.4cm}]{{\sc{
Encoding Step}, repeat for $ i \in \{1,\dots ,\frac{N}{M}\}$}}& \qw & \qw & \qw &\qw & \qw &\qw & \gategroup[8,steps=3,style={dashed,rounded
corners,fill=gate_group_color_a, inner
xsep=2pt},background,label style={label
position=above,anchor=north,yshift=0.4cm}]{{\sc
Final Step}} 
&\gate[6,style={fill=gate_color_xor}][1cm]{\text{XOR}}
&\qw&\qw\\
    \lstick{$\text{I}_1 \ket{0}$} &\qw& \gate[style={fill=gate_color_X}]{X} & \qw & \qw&\gate[style={fill=gate_color_im}]{iM+1} & \gate[2,style={fill=gate_color_and}][1cm]{\text{AND } } & \gate[6,style={fill=gate_color_mtc}]{MTC_i} & \qw & \gate[6,style={fill=gate_color_mtc}]{MTC_i^+} & \gate[2,style={fill=gate_color_and}][1cm]{\text{AND }} & \gate[style={fill=gate_color_im}]{iM+1}  &\qw&\qw& \gate[style={fill=gate_color_X}]{X}&\qw\\
    \lstick{$\text{C}_1 \ket{0}$} &\qw& \qw & \qw & \qw &\qw& \qw & \qw & \qw & \qw & \qw & \qw &\qw&\qw&\qw&\qw\\
    \lstick{$\vdots$} &\qw& \vdots & \qw & \qw&\vdots & \vdots & \qw & \qw & \qw & \vdots & \vdots &\qw& \qw&\vdots&\qw&\qw\\
    \lstick{$\text{I}_M \ket{0}$} &\qw& \gate[style={fill=gate_color_X}]{X} & \qw &\qw& \gate[style={fill=gate_color_im}]{iM+M} & \gate[2,style={fill=gate_color_and}][1cm]{\text{AND } }& \qw & \qw & \qw & \gate[2,style={fill=gate_color_and}][1cm]{\text{AND } } & \gate[style={fill=gate_color_im}]{iM+M} &\qw&\qw&\gate[style={fill=gate_color_X}]{X}&\qw\\
    \lstick{$\text{C}_M \ket{0}$} &\qw& \qw & \qw &\qw& \qw & \qw & \qw & \qw & \qw & \qw & \qw &\qw&\qw&\qw&\qw\\
    \lstick{$\text{CTRL} \ket{0}$} & \qw & \qw & \qw & \qw &\qw& \qw & \qw & \gate[2,style={fill=gate_color_crot}][0.5cm]{\text{Ctrl Rot}} & \qw & \qw &\qw&\qw&\qw&\qw&\qw\\
    \lstick{$\text{FLAG} \ket{0}$} & \qw & \qw & \qw & \qw &\qw& \qw & \qw & \qw & \qw& \qw &\qw&\qw&\qw&\qw&\qw
\end{quantikz}

%% file: circuit_tex_files/initial_step.tex
\input{colors}

\begin{quantikz}[row sep={0.8cm,between origins}, column sep=0.15cm]
    \lstick{$\text{SYS} \ket{0}$} &\gategroup[8,steps=3,style={dashed,rounded
corners,fill=gate_group_color_a, inner
xsep=2pt},background,label style={label
position=above,anchor=north,yshift=0.4cm}]{{\sc
Initial Step}} & \gate[style={fill=gate_color_H}]{H} & \gate[6,style={fill=gate_color_xor}]{\text{XOR}} &\qw \\
    \lstick{$\text{I}_1 \ket{0}$} &\qw& \gate[style={fill=gate_color_X}]{X} & \qw & \qw\\
    \lstick{$\text{C}_1 \ket{0}$} &\qw& \qw & \qw & \qw\\
    \lstick{$\vdots$} &\qw& \vdots & \qw & \qw\\
    \lstick{$\text{I}_M \ket{0}$} &\qw& \gate[style={fill=gate_color_X}]{X} & \qw &\qw\\
    \lstick{$\text{C}_M \ket{0}$} &\qw& \qw & \qw &\qw\\
    \lstick{$\text{CTRL} \ket{0}$} & \qw & \qw & \qw & \qw\\
    \lstick{$\text{FLAG} \ket{0}$} & \qw & \qw & \qw & \qw
\end{quantikz}

%% file: circuit_tex_files/XOR.tex
\tikzset{
    operator/.append style={draw,fill=gate_color_xor}  }
\begin{quantikz}[column sep = 4pt]
\lstick{$\ket{\text{Control}}$} & \gate[wires=5]{\text{XOR}} & \qw \\
\lstick{$\ket{\text{Targ}_1}$} &                       & \qw \\
\lstick{$\ket{\text{Targ}_2}$} &                       & \qw \\
\vdots                          \\
\lstick{$\ket{\text{Targ}_{n}}$} &                       & \qw \\
\end{quantikz}
\quad = \quad
% Decomposed Circuit
\begin{quantikz}[column sep = 4pt]
\lstick{} & \ctrl{1} & \ctrl{2}     & \dots    & \ctrl{4}& \qw  \\
\lstick{} & \targ    & \qw & \qw      & \dots&\qw& \qw \\
\lstick{}   &\qw          & \targ    & \qw&\dots &\qw & \qw \\
\vdots \\
\lstick{} & \qw      & \qw      & \dots      & \targ& \qw & \qw  \\
\end{quantikz}

%% file: circuit_tex_files/AND.tex
\begin{quantikz}
\lstick{$\ket{\text{Control}_1}$} & \gate[wires=6,style={fill=gate_color_and}]{\text{AND}} & \qw \\
\lstick{$\ket{\text{Control}_2}$} &                       & \qw \\
\lstick{} &                       & \qw \\
\vdots  \\
\lstick{$\ket{\text{Control}_{n}}$} &                       & \qw \\
\lstick{$\ket{\text{Targ}}$} &                       & \qw \\
\end{quantikz}
\quad = \quad
% Decomposed Circuit
\begin{quantikz}
\lstick{} & \ctrl{1} & \qw\\
\lstick{} & \ctrl{1}    & \qw\\
\lstick{}   &\ctrl{2}         & \qw\\
\vdots \\
\lstick{} & \ctrl{1}       & \qw\\
\lstick{} & \targ    & \qw&\qw\\
\end{quantikz}

%% file: circuit_tex_files/MTC_i.tex
\begin{quantikz}[column sep = 4pt]
\lstick{$\ket{\text{C}_1}$} & \gate[wires=6,style={fill=gate_color_mtc}]{\text{MTC}} & \qw \\
\lstick{$\ket{\text{C}_2}$} &                       & \qw \\
\lstick{} &                       & \qw \\
\vdots                          \\
\lstick{$\ket{\text{C}_{N}}$} &                       & \qw \\
\lstick{$\ket{\text{CTRL}}$} &                       & \qw \\
\end{quantikz}
\quad = \quad
% Decomposed Circuit
\begin{quantikz}[column sep = 4pt]
\lstick{} & \ctrl{5} & \qw&\qw& \qw& \qw&\qw\\
\lstick{} & \qw    & \ctrl{4}&\qw& \qw& \qw &\qw\\
\lstick{}   &\qw&\qw  &\ctrl{3}      & \qw& \qw&\qw\\
\vdots \\
\lstick{} & \qw       & \qw&\qw  & \dots & \ctrl{1}&\qw\\
\lstick{} & \gate[wires=1,style={fill=gate_color_mtc}]{B_{0}}  &  \gate[wires=1,style={fill=gate_color_mtc}]{B_{1}}& \gate[wires=1,style={fill=gate_color_mtc}]{B_{2}} & \dots &  \gate[wires=1,style={fill=gate_color_mtc}]{B_{N-1}}&\qw\\
\end{quantikz}

%% file: circuit_tex_files/Memory2Control.tex
\begin{quantikz}
\lstick{$\ket{\text{C}_1}$} & \gate[wires=11,style={fill=gate_color_mtc}]{\text{Mem2CTRL}} & \qw \\
\lstick{$\ket{\text{C}_2}$} &                       & \qw \\
\lstick{$\ket{\text{C}_3}$}  & \qw \\
\lstick{$\ket{\text{C}_4}$} &                       & \qw \\
\vdots                          \\
\lstick{$\ket{\text{C}_{N-1}}$} &                       & \qw \\
\lstick{$\ket{\text{C}_{N}}$} &                       & \qw \\
\lstick{$\ket{\text{Targ}_1}$} &                       & \qw \\
\lstick{$\ket{\text{Targ}_2}$} &                       & \qw \\
\vdots \\
\lstick{$\ket{\text{Targ}_\frac{N}{2}}$} &                       & \qw \\
\end{quantikz}
\quad = \quad
% Decomposed Circuit
\begin{quantikz}
\lstick{} & \ctrl{7} & \qw&\qw& \qw& \dots&  \qw&  \qw&  \qw\\
\lstick{} & \qw    & \ctrl{6}&\qw& \qw& \dots &  \qw&  \qw&  \qw\\
\lstick{}   &\qw&\qw  &\ctrl{6}      & \qw& \dots&  \qw&  \qw&  \qw\\
\lstick{}   &\qw&\qw  &\qw   & \ctrl{5}& \dots&  \qw&  \qw&  \qw\\
\vdots \\
\lstick{} & \qw       & \qw&\qw & \qw & \dots & \ctrl{5}&  \qw&  \qw\\
\lstick{} & \qw       & \qw&\qw & \qw & \dots &\qw & \ctrl{4}&  \qw\\
\lstick{} & \gate[wires=1,style={fill=gate_color_mtc}]{B_{0}}  &  \gate[wires=1,style={fill=gate_color_mtc}]{B_{1}}& \qw & \qw &  \qw&  \qw&  \qw&  \qw\\
\lstick{} & \qw  &  \qw& \gate[wires=1,style={fill=gate_color_mtc}]{B_{2}} &\gate[wires=1,style={fill=gate_color_mtc}]{B_{3}}& \dots &  \qw &  \qw&  \qw\\
\vdots\\
\lstick{} & \qw  &  \qw& \qw & \qw & \dots &  \gate[wires=1,style={fill=gate_color_mtc}]{B_{N-2}}&  \gate[wires=1,style={fill=gate_color_mtc}]{B_{N-1}}&  \qw\\
\end{quantikz}

%% file: circuit_tex_files/Ctrl_comp.tex
\begin{quantikz}
\lstick{$\ket{\text{Targ}_1}$} & \gate[wires=5,style={fill=gate_color_mtc}]{\text{CTRL Comp}} & \qw \\
\lstick{$\ket{\text{Targ}_2}$} &                       & \qw \\
\lstick{$\ket{\text{Targ}_3}$} &                       & \qw \\
\vdots                          \\
\lstick{$\ket{\text{Targ}_{\frac{N}{2}}}$} &                       & \qw \\
\end{quantikz}
\quad = \quad
% Decomposed Circuit
\begin{quantikz}
\lstick{$\ket{\text{Targ}_1}$} & \ctrl{1} & \qw&\qw& \qw& \qw&\qw\rstick[wires=2]{2}&\qw\rstick[wires=4]{4}&\qw\rstick[wires=8]{8} &\dots&\qw \rstick[wires=10]{ $\quad \frac{N}{2}$} \\
\lstick{$\ket{\text{Targ}_2}$} & \targ{}   & \ctrl{2}&\qw& \qw& \qw &\qw&\qw&\qw&\dots&\qw\\
\lstick{$\ket{\text{Targ}_3}$}   &\ctrl{1}&\qw  &\qw      & \qw& \qw&\qw\rstick[wires=2]{2}&\qw&\qw&\dots&\qw\\
\lstick{$\ket{\text{Targ}_4}$}   &\targ{} &\targ{}  &\ctrl{4}    & \qw& \qw&\qw&\qw&\qw&\dots&\qw\\
\lstick{$\ket{\text{Targ}_5}$} & \ctrl{1} & \qw&\qw& \qw& \qw&\qw\rstick[wires=2]{2}&\qw\rstick[wires=4]{4}&\qw&\dots&\qw\\
\lstick{$\ket{\text{Targ}_6}$} & \targ{}   & \ctrl{2}&\qw& \qw& \qw &\qw&\qw&\qw&\dots&\qw\\
\lstick{$\ket{\text{Targ}_7}$}   &\ctrl{1}&\qw  &\qw      & \qw& \qw&\qw\rstick[wires=2]{2}&\qw&\qw&\dots&\qw\\
\lstick{$\ket{\text{Targ}_8}$}   &\targ{} &\targ{}  &\targ{}    & \qw& \qw&\qw&\qw&\qw&\dots&\ctrl[dashed lines]{2}\\
\vdots\\
\lstick{$\ket{\text{Targ}_\frac{N}{2}}$} & \qw & \qw & \qw & \qw &\qw&\qw&\qw&\qw&\dots&\targ{}\\
\end{quantikz}

%% file: circuit_tex_files/controlled_rotations.tex
\begin{quantikz}
\lstick{$\ket{\text{CTRL}_1}$} & \gate[wires=5,style={fill=gate_color_crot}]{\text{Ctrl Rot}} & \qw \\
\lstick{$\ket{\text{CTRL}_2}$} &                       & \qw \\
\vdots                          \\
\lstick{$\ket{\text{CTRL}_{L}}$} &                       & \qw \\
\lstick{$\ket{\text{FLAG}}$} &                       & \qw \\
\end{quantikz}
\quad = \quad
% Decomposed Circuit
\begin{quantikz}
\lstick{} & \ctrl{4} & \qw& \dots& \qw&\qw\\\
\lstick{} & \qw    &  \ctrl{3}& \dots& \qw&\qw\\\
\vdots \\
\lstick{} & \qw       & \qw& \dots& \ctrl{1}&\qw\\
\lstick{} & \gate[style={fill=gate_color_crot}]{R_y(2\pi)}   & \gate[style={fill=gate_color_crot}]{R_y(\frac{\pi}{2})} &\dots &\gate[style={fill=gate_color_crot}]{R_y(\frac{\pi}{2^{L-1}})}& \qw\\
\end{quantikz}